\documentclass[final,onefignum,onetabnum]{siamart171218}

\usepackage{amsfonts}
\usepackage{graphicx}
\usepackage{epstopdf}
\usepackage{algorithmic}
\usepackage[ruled,algo2e, vlined, linesnumbered]{algorithm2e}

\ifpdf
  \DeclareGraphicsExtensions{.eps,.pdf,.png,.jpg}
\else
  \DeclareGraphicsExtensions{.eps}
\fi

\newsiamremark{remark}{Remark}
\newsiamremark{hypothesis}{Hypothesis}
\crefname{hypothesis}{Hypothesis}{Hypotheses}
\newsiamthm{claim}{Claim}
\newsiamremark{algo}{Algorithm}

\headers{Hardness and ease}{J. Klassen, M. Marvian, S. Piddock, M. Ioannou, I. Hen, and B.M. Terhal}

\title{Hardness and ease of curing the sign problem for two-local qubit Hamiltonians
\thanks{
\funding{JK, MI and BT acknowledge funding through ERC grant EQEC No. 682726. MM and IH are partially supported by the Office of the Director of National Intelligence (ODNI), Intelligence Advanced Research Projects Activity (IARPA), via the U.S. Army Research Office contract W911NF-17-C-0050. IH is partially supported by the Air Force Research laboratory under agreement number FA8750-18-1-0044. The U.S. Government is authorized to reproduce and distribute reprints for Governmental purposes notwithstanding any copyright notation thereon. The views and conclusions contained herein are those of the authors and should not be interpreted as necessarily representing the official policies or endorsements, either expressed or implied, of the ODNI, IARPA, or the U.S. Government.}}
}

\author {Joel Klassen~\footnotemark[2]~\footnotemark[3]
\and Milad Marvian~\footnotemark[2]~\footnotemark[4]  
\and Stephen Piddock~\footnotemark[5]
\and Marios Ioannou~\footnotemark[6] 
\and Itay Hen~\footnotemark[8] 
\and Barbara M. Terhal~\footnotemark[9] 
}

\usepackage{amsopn}

\usepackage{braket}
\usepackage{empheq}

\newcommand{\bes} {\begin{subequations}}
\newcommand{\ees} {\end{subequations}}
\newcommand{\ba}{\begin{eqnarray}}
\newcommand{\ea}{\end{eqnarray}}

\newcommand{\lowrank}{rank-1}
\newcommand{\highrank}{rank$>$1}
\newcommand{\capHRCC}{Rank$>$1 Connected Component}
\newcommand{\HRCC}{rank$>$1 connected component}

\newcommand{\abrvHRCC}{RCC}

\begin{document}
\maketitle

\renewcommand{\thefootnote}{\fnsymbol{footnote}}
\footnotetext[2]{Klassen and Marvian contributed equally to this work.}
\footnotetext[3]{QuTech, Delft University of Technology
  (\email{J.D.Klassen@tudelft.nl})}
\footnotetext[4]{Research Laboratory of Electronics, MIT
  (\email{mmarvian@mit.edu})}
  \footnotetext[5]{School of Mathematics, University of Bristol
  (\email{stephen.piddock@bristol.ac.uk})}
\footnotetext[6]{Faculty of Physics, Ludwig Maximilian University of Munich (\email{M.Ioannou@tudelft.nl})}
\footnotetext[8]{Information Sciences Institute, University of Southern California (\email{itayhen@isi.edu})}
\footnotetext[9]{QuTech, Delft University of Technology(\email{B.M.Terhal@tudelft.nl})}

\renewcommand{\thefootnote}{\arabic{footnote}}

\begin{abstract}
We examine the problem of determining whether a multi-qubit two-local Hamiltonian can be made stoquastic by single-qubit unitary transformations.  We prove that when such a Hamiltonian contains one-local terms, then this task can be NP-hard. This is shown by constructing a class of Hamiltonians for which performing this task is equivalent to deciding $3$-SAT. In contrast, we show that when such a Hamiltonian contains no one-local terms then this task is easy, namely we present an algorithm which decides, in a number of arithmetic operations over $\mathbb{R}$ which is polynomial in the number of qubits, whether the sign problem of the Hamiltonian can be cured by single-qubit rotations. 
\end{abstract}



\section{Introduction}

The sign problem in quantum physics has long been recognized as one of the main 
impediments of efficient Monte-Carlo simulation of quantum many-body systems~\cite{sign:troyer,Loh-PRB-90}. Hamiltonians that do not suffer from the sign problem have recently been given the name `stoquastic'~\cite{Bravyi:QIC08}, a term 
which aims to capture the relationship between these Hamiltonians and stochastic processes. Many interesting quantum models such as the transverse field Ising model, the Bose-Hubbard model,
and a collection of kinetic particles in a position-dependent potential, are stoquastic. However, stoquasticity, as introduced in Ref.~\cite{Bravyi:QIC08}, is a basis-dependent concept. It 
requires that the Hamiltonian of the physical model in question be real and have non-positive off-diagonal elements in a given basis. For a many-body local Hamiltonian acting on $n$ qubits, this basis is typically a product basis on which the terms of the Hamiltonian act locally and can be efficiently described.
The non-positivity of the off-diagonal elements of a stoquastic Hamiltonian matrix in a particular basis has important consequences. It guarantees, via the Perron-Frobenius theorem, that there exists a set of orthonormal states, spanning the ground state subspace, whose amplitudes are non-negative in that basis \cite{Bravyi:2009sp}. In addition, the quantum partition function of a stoquastic Hamiltonian can be expressed as a sum of non-negative, easily computable, weights, which implies that Markov chain Monte-Carlo algorithms can be used to perform importance sampling of the quantum configuration space to calculate thermal averages of physical observables, using these weights as (unnormalized) probabilities. For this reason, it is said that stoquastic Hamiltonians do not suffer from the sign problem \cite{Bravyi:2009sp, Bravyi:QIC08}. However it is important to note that the absence of a sign problem does not necessarily imply polynomial-time convergence of standard Monte-Carlo methods \cite{Hastings2013, Jarret2016, Bringewatt2018}.

From a computational complexity perspective, the problem of estimating ground-state energies of stoquastic local Hamiltonians is considered easier than for general Hamiltonians~\cite{Bravyi:QIC08, bravyi:2006aa}. Moreover, in the classification of the complexity of estimating ground-state energies of local Hamiltonians, stoquastic Hamiltonians appear as the only intermediate class between classical Hamiltonians and general Hamiltonians~\cite{Cubitt:2016vl}. 
Stoquastic local Hamiltonians are of interest not only in quantum complexity theory. In Ref.~\cite{Bravyi:2009sp} it was shown that deciding whether a stoquastic Hamiltonian is frustration-free is a MA-complete problem. Recently Ref.~\cite{AG:stoquastic} showed that the gapped version of this question is in NP, linking derandomization of MA to NP to the possibility of gap amplification of stoquastic local Hamiltonians.


Identifying classes of Hamiltonians that are stoquastic is clearly motivated both from practical and complexity-theoretic perspectives. Given that stoquasticity is basis-dependent, an interesting question arises: under what circumstances can the sign problem be `cured', as coined in Ref.~\cite{marvian2018computational}, by performing local basis changes? This is the main question explored in this paper.

It is worth noting that the sign problem may be resolved by means other than a local basis transformation. Other methods for generating positive-valued decompositions of the partition function include, e.g., re-summation techniques wherein negative-valued weights in the decomposition are grouped together with positive ones to form positive `super-weights' that can in turn be treated as probabilities in a quantum Monte Carlo algorithm~\cite{resum1,resum2,hen2019}. Other methods also include applying a constant-depth quantum circuit \cite{torlai+:pos}. These other methods are beyond the scope of this paper. 


Naturally, devising techniques for obviating or mitigating the sign problem has been a focus of much research in the quantum Monte-Carlo (QMC) community since its inception \cite{Suzuki1993, Landau:2005:GMC:1051461,resum1,resum2,hen2019}. In particular, the importance of basis choice has been widely recognized (see, e.g. ~\cite{Bishop2001, Hatano1992, Hastings2015, Ringel2017}). Recognizing the key role that stoquastic Hamiltonians play both in computational complexity and in physics, a more general algorithmic approach has recently been launched to determine whether a Hamiltonian can be made stoquastic~\cite{ marvian2018computational, klassen2018two, Bausch2018}. In this 
paper we present an important strengthening of these initial results.  

Stoquasticity has also attracted attention from the experimental community. In particular there has been  a growing interest in engineering Hamiltonian interactions that are not stoquastic~\cite{Vinci:2017aa,Ozfidan2019}. Some of the reasons for this include: enhancing the performance of quantum annealer protocols for optimization \cite{nonstoq1,nonstoq2,albash}; realizing universal adiabatic quantum computers~\cite{AL:review,aharonov_adiabatic_2007,Biamonte:07}; and physically emulating quantum many-body systems~\cite{Deng2019}. Here too, the question of whether and how local basis changes can cure the sign problem is highly relevant, as experimental quantum advantages hinge on the inability to simulate non-stoquastic interactions on classical computers.


\section{Previous work}

In what follows, we will refer to Hermitian matrices that are real and have only non-positive off-diagonal elements as \textbf{symmetric $Z$-matrices}~\cite{BP:book}.

A no-go result presented recently by some of the authors of this paper states that the problem of determining whether there exists a sign-curing transformation for general local Hamiltonians is NP-hard when one is restricted to applying particular single-qubit transformations  to the Hamiltonian \cite{marvian2018computational}. This result can be summarized as

\begin{theorem}\cite{marvian2018computational}
Let $H$ be a three-local $n$-qubit Hamiltonian and let {\sf LocalCliffordSignCure} be the problem of determining whether there exist single-qubit Clifford transformations ${\sf C}_u$ with ${\sf C}=\bigotimes_{u=1}^n {\sf C}_u$ such that ${\sf C} H {\sf C}^{\dagger}$ is a symmetric $Z$-matrix. {\sf LocalCliffordSignCure} is {\rm NP}-hard. Let $H$ be a $6$-local $n$-qubit Hamiltonian and let {\sf LocalRealRotSignCure} be the problem of determining whether there exist real single-qubit rotations $R_u \in SO(2)$ with $R=\bigotimes_{u=1}^n R_u$ such that $R H R^{T}$ is a symmetric $Z$-matrix. {\sf LocalRealRotSignCure} is {\rm NP}-hard.
\label{theo:NPprev}
\end{theorem}

{\em Remark}: When dealing with $k$-local Hamiltonians with $k> 2$, it is important to note that two distinct notions of stoquasticity have been defined in Ref.~\cite{Bravyi:QIC08}, namely there exist termwise-stoquastic Hamiltonians and globally-stoquastic Hamiltonians. A globally-stoquastic Hamiltonian is a symmetric $Z$-matrix, while a Hamiltonian which is $k$-local termwise-stoquastic is one which can be decomposed into $k$-local terms such that each term is a symmetric $Z$-matrix. A globally-stoquastic Hamiltonian need not be termwise-stoquastic while a termwise-stoquastic Hamiltonian is always globally-stoquastic. The results in Theorem~\ref{theo:NPprev} hold for both definitions. For the two-local Hamiltonians in this paper one can prove~\cite{Bravyi:QIC08} that these notions coincide, hence we do not distinguish between these two definitions in this paper. We provide a proof of this equivalence in Proposition~\ref{prop:global2} for completeness. 

It was also recently shown, by other authors of this paper, that for a particularly broad family of two-local Hamiltonians, namely arbitrary XYZ Heisenberg Hamiltonians, there is an efficient procedure for determining whether the sign problem can be cured by single-qubit unitary transformations:

\begin{theorem}\cite{klassen2018two}
Let $H=\sum_{u,v} H_{uv}$, $u=1, \ldots, n$, $v=1,\ldots, n$ be an $n$-qubit Hamiltonian with $H_{uv}=a^{uv}_{XX} X_u X_v+a^{uv}_{YY}Y_u Y_v +a^{uv}_{ZZ}Z_u Z_v$, where each $a^{uv}_{kk}$ is given with $O(1)$ bits. There is an efficient algorithm, which we call the \textbf{XYZ-algorithm}, that runs in time $O(n^3)$ to decide whether there are single-qubit rotations $U_u \in SU(2)$ with $U=\bigotimes_{u=1}^n U_u$ such that $U H U^{\dagger}$ is a symmetric $Z$-matrix.
\label{theo:KT}
\end{theorem}

An essential step in the proof of Theorem~\ref{theo:KT} was to show that single-qubit Clifford transformations suffice as basis changes, reducing the problem to an optimization problem over a discrete set of degrees of freedom.  

\section{Main results}

This paper aims to bridge the gap between these two previous results, and identify the boundary between classes of Hamiltonians for which curing the sign problem by local basis transformations is hard and those for which this problem is easy.  The main results of this paper address the following problem. 
\begin{definition}[{\sf LocalSignCure}]
Given a two-local $n$-qubit Hamiltonian, {\sf LocalSignCure} is the problem of determining whether there exists a set of single-qubit unitary transformations $U_u \in SU(2)$ with $U=\bigotimes _{u=1}^n U_u$ such that $U HU^{\dagger}=\tilde{H}$ is a symmetric $Z$-matrix.
\end{definition}
We colloquially refer to such unitary transformation $U$ as a \textbf{sign-curing transformation}, and say that the sign problem of a Hamiltonian can be \textbf{cured} if such a transformation exists.


The main results of this paper are the following two theorems and constitute a strengthening of Theorems~\ref{theo:NPprev} and~\ref{theo:KT} to two-local Hamiltonians. 
In Section~\ref{sec:hardnessResult} we will prove the following theorem.

\begin{theorem}
There exists a family of a two-local $n$-qubit Hamiltonians for which {\sf LocalSignCure} is {\rm NP}-complete. 
\label{thm1}
\end{theorem}

 To prove this, we modify the constructions introduced in~Ref.~\cite{marvian2018computational} thereby reducing the locality of Hamiltonians from three-local (in the case of the single-qubit Clifford group) and six-local (in the case of the single-qubit orthogonal group) to two-local. This result demonstrates that {\sf LocalSignCure} is hard in general. Theorem~\ref{thm1} additionally demonstrates that deciding if a multi-qubit two-local Hamiltonian can be sign-cured by single-qubit Clifford transformations is hard. We show in Appendix~\ref{app:cliffordsEasy} that, in the absence of one-local terms, this task is easy. We should stress here that it is not clear that {\sf LocalSignCure} is a problem in NP for general Hamiltonians. We expand on this point in the Discussion section~\ref{sec:dis}. 

For a relatively broad subclass of two-local  Hamiltonians we can, however, show that finding local basis changes is easy.



\begin{theorem}
Let $H$ be an \textbf{exactly two-local} $n$-qubit Hamiltonian, meaning a Hamiltonian of the form $H=\sum_{u,v} H_{uv}$ with $H_{uv} = \sum_{k, l \in \{X,Y,Z\}} (\beta_{uv})_{kl} \sigma_k^u \otimes \sigma_l^v $ with $\sigma_k^u$ a Pauli matrix of type $k$, acting on qubit $u$, and $(\beta_{uv})_{kl}$ is given with $O(1)$ bits.  There is an efficient algorithm, using $O(n^3)$ arithmetic operations over $\mathbb{R}$, which solves {\sf LocalSignCure} for $H$.
\label{thm2}
\end{theorem}

This algorithm is presented in Section~\ref{sec:curingResult}.  It employs the XYZ-algorithm referred to in Theorem~\ref{theo:KT}, as a subroutine. It is important to note that, just as in the XYZ algorithm, this algorithm makes no guarantee that $H$ can be cured, it only efficiently \textbf{decides} whether or not $H$ can be cured. 
An important difference between the new algorithm in Theorem~\ref{thm2} and the XYZ-algorithm is that the new algorithm requires finding singular value decompositions of matrices specified by $O(1)$ bits, as well as intersections of vector subspaces, while the XYZ-algorithm required solving a discrete optimization problem. Since we do not address the question of how a finite-precision implementation of these standard linear algebra operations affects the accuracy with which we decide whether the sign problem of $H$ can be cured, we state our theorem in terms of arithmetic operations over $\mathbb{R}$. However the algorithm is expected to be numerically stable insofar as repeated composition of orthogonal rotations, and intersections of vector spaces, are numerically stable. A rigorous account of the complexity of this problem would require a finite precision formulation. We will not attempt to do this here.



The upshot of these results is that the presence of local fields can change the complexity class of curing the sign problem of two-local Hamiltonians by single-qubit unitaries from P to NP-complete.

\section{Preliminaries}

For ease of exposition and reference we start by stating the following observation about two-qubit Hamiltonians. 

\begin{proposition}
A two-qubit Hamiltonian $H=\sum_{k,l=I,X,Y,Z} a_{kl} \sigma_k \otimes \sigma_l$ is a symmetric $Z$-matrix if and only if $a_{IY}=a_{YI}=a_{XY}=a_{YX}=a_{ZY}=a_{YZ}=0$ (the matrix is real) and $a_{XX} \leq - \vert a_{YY} \vert$ and $a_{IX} \leq -\vert a_{ZX} \vert$, $a_{XI} \leq -\vert a_{XZ} \vert$ (the matrix has non-positive off-diagonal elements).
\label{prop:basic}
\end{proposition}


%
\begin{proposition}\label{prop:ortho} \cite{klassen2018two}
Given a two-qubit Hamiltonian $H=\sum_{k,l=I,X,Y,Z} a_{kl} \sigma_k \otimes \sigma_l$, where the two-local term can be concisely represented by the $3 \times 3$ matrix \begin{equation*}\beta = \left( \begin{matrix}
a_{XX} & a_{XY} & a_{XZ} \\
a_{YX} & a_{YY} & a_{YZ} \\
a_{ZX} & a_{ZY} & a_{ZZ}
\end{matrix} \right).
\end{equation*} A pair of single-qubit unitary transformations $U_1$ and $U_2$ with action: $H \rightarrow (U_1\otimes U_2)H(U_1\otimes U_2)^{\dagger}$ corresponds to a pair of $SO(3)$ rotations $O_1$, $O_2$ acting on the $\beta$-matrix: $\beta \rightarrow O_1^{\sf T} \beta O_2$.
\end{proposition}
For the curious reader, an example of a Hamiltonian which is not stoquastic under any single-qubit unitary transformations is provided in Appendix~\ref{app:simpleExample}. 

It was claimed in Ref.~\cite{Bravyi:QIC08}, without proof, that a two-local termwise-stoquastic Hamiltonian with respect to a basis is also globally stoquastic. We include the proof here:

\begin{proposition}\cite{Bravyi:QIC08}
A two-local Hamiltonian $H$ acting on $n$ qubits is a symmetric $Z$-matrix in the computational basis if and only if $H=\sum_{u<v} D_{uv}$ where each $D_{uv}$ acts nontrivially on at most two qubits, namely qubits $u$ and $v$, and $D_{uv}$ is a symmetric $Z$-matrix.
\label{prop:global2}
\end{proposition}

\begin{proof}
Let $\ket{x}$, with $x \in \{0,1\}^n$, denote a computational basis state. If there exists a decomposition $H=\sum_{u<v} D_{uv}$ such that $D_{uv}$ is real and $\forall x \neq y$, $\bra{x} D_{uv} \ket{y}\leq 0$, then $H$ is real and $\forall x \neq y$, $\bra{x} H \ket{y} = \sum_{u,v} \bra{x} D_{uv} \ket{y} \leq 0$.
This proves one direction of the bi-conditional, we now prove the other direction. Since $H$ is real, $H=H^{\sf T}$. Therefore every Pauli operator $P$ in the Pauli expansion of $H$ must satisfy $P=P^{\sf T}$, and so $H$ does not contain any Pauli operators with odd numbers of $Y$ terms.  Let $d_H(x,y)$ denote the Hamming distance between bit strings $x$ and $y$. Since $H$ is two-local, $H=M^{(0)}+M^{(1)}+M^{(2)}$ where $\bra{x} M^{(m)}\ket{y}=0$ whenever $d_H(x,y)\neq m$. In other words the Hamiltonian decomposes into three sets: $M^{(0)}$ contains all terms which are diagonal (i.e., terms of the form $ZI$, $IZ$ and $ZZ$), $M^{(1)}$ contains all terms that flip 1 bit (i.e., of the form $XZ$, $ZX$, $XI$, $IX$), and $M^{(2)}$ contains all terms that flip two bits (of the form $XX$ and $YY$). There is no particular condition which has to be fulfilled for the diagonal group $M^{(0)}$, and so we ignore it. 
Furthermore, from the condition $\forall x\neq y$ $\bra{x} H \ket{y} \leq 0$, it follows that $\forall x\neq y$ $\bra{x} M^{(1)} \ket{y} \leq 0$ and $\bra{x} M^{(2)} \ket{y} \leq 0$, since $M^{(1)}$ and $M^{(2)} $ are non-zero at different off-diagonal positions.



For any potential decomposition $H=\sum_{u<v} D_{uv}$ we can similarly write $D_{uv}=D^{(0)}_{uv}+D^{(1)}_{uv}+D^{(2)}_{uv}$, grouping diagonal, one-qubit flipping, and two-qubit flipping terms.  Since $M^{(m)}$ contains all terms which flip $m$-qubits, $M^{(m)}=\sum_{u,v} D^{(m)}_{uv}$. In the case of $m=2$, $D^{(2)}_{uv}$ and $D^{(2)}_{wr}$ are non-zero at different off-diagonal positions when $u,v\neq w,r$, and so $\forall x\neq y$ $\bra{x} M^{(2)} \ket{y} \leq 0$ implies $\forall x\neq y, \forall u<v, \bra{x} D^{(2)}_{uv} \ket{y} \leq 0$. 

In the case of $m=1$, $D^{(1)}_{uv}$ and $D^{(1)}_{wx}$ may both be non-zero on the same off-diagonal position, and so we must use a different argument. We can write $M^{(1)}=\sum_{u,v\colon u < v} [a_{XZ}^{uv} X_u Z_v +a_{ZX}^{uv} Z_u X_v]+\sum_u a_X^u X_u$. By writing out matrix elements one can show that 
\begin{equation*}
\forall x\neq y, \bra{x} M^{(1)} \ket{y} \leq 0 \Rightarrow \forall u \;\; a_X^u+\sum_{v\colon v > u} \Delta^v a_{XZ}^{uv}+\sum_{w\colon w<u} \Delta^w a_{ZX}^{wu} \leq 0,
\end{equation*}
for all choices of sign-patterns $\Delta^u=\pm 1$. Note that $\Delta^u=\pm 1$ since $Z_u$ is applied on the identical $u$th bit in $x$ and $y$ which can either be 0 or 1. This implies that $\forall u$ we have 
\begin{equation}
a_X^u \leq -\left(\sum_{v\colon v > u} |a_{XZ}^{uv}|+\sum_{w\colon w<u} |a_{ZX}^{wu}|\right).
\label{eq:global}
\end{equation}
A local term is of the form $D_{uv}^{(1)}=a_{XZ}^{uv} X_u Z_v+a_{ZX}^{uv} Z_u X_v+a_{XI}^{uv} X_u I_v+a_{IX}^{uv} I_u Z_v$, where the coefficients $a_{XI}^{uv}, a_{IX}^{uv}$ can be freely chosen up to the overall constraint 
$a_X^u=\sum_{v\colon v > u} a_{XI}^{uv}+\sum_{w\colon w < u} a_{IX}^{wu}$. Now, clearly, if Eq.~(\ref{eq:global}) holds, then one can always distribute $a_X^u$ into a sum over $a_{XI}^{uv}$ (for $v > u$) and $a_{IX}^{wu}$ (for $w < u$) such that each $a_{XI}^{uv} \leq -|a_{XZ}^{uv}|$ and each $a_{IX}^{wu} \leq -|a_{ZX}^{wu}|$. Hence, by Proposition~\ref{prop:basic} there is a decomposition with terms $D_{uv}$ such that $D_{uv}^{(1)}$ is a symmetric $Z$-matrix, and so $D_{uv}$ is a symmetric $Z$-matrix.
\end{proof}

\section{{\sf LocalSignCure} for a class of two-local Hamiltonians is NP-complete}\label{sec:hardnessResult}


In this section we present a family of Hamiltonians for which solving {\sf LocalSignCure} is NP-complete, and thus show that {\sf LocalSignCure} is NP-hard. 

We will first show that {\sf LocalSignCure} for this class of Hamiltonians is in NP. This is not immediately apparent, since local basis transformations have a continuous parametrization, hence one either has to allow for approximate sign-curing transformations or prove that for this particular class of Hamiltonians any sign-curing transformation is a member of a discrete subset of transformations. We settle this problem by proving in Lemma~\ref{lem:restrictedbasis} that with the addition of ancilla qubits and ``gadget" interactions, any Hamiltonian in this class can be converted into one for which any sign-curing transformation must consist of either Hadamard gates or the identity operation.
In order to prove that the problem is NP-hard, we show how to encode any 3-SAT instance into the problem of curing a corresponding Hamiltonian using the identity or Hadamard gates. In Lemma~\ref{lem:H3SAT} we prove that such a curing transformation exists if and only if the corresponding 3-SAT instance is satisfiable.
A proof of Theorem~\ref{thm1} follows straightforwardly by considering  {\sf LocalSignCure} for the family of Hamiltonians constructed by adding the gadgets (Lemma~\ref{lem:restrictedbasis}) to the Hamiltonians corresponding to 3-SAT instances (Lemma~\ref{lem:H3SAT}).

\subsection{Hadamard sign curing gadget}
In this section we introduce the ``gadget" interactions which will effectively force any sign-curing transformation to be from a discrete subset of transformations. Let $W_u$ be a single-qubit Hadamard on qubit $u$: this is a convention we will use throughout this section and the next.

\begin{lemma}
	\label{lem:restrictedbasis}
Let $H$ be a two-local Hamiltonian on $n$ qubits. For each qubit $u \in \{1,\dots n\}$, add three ancilla qubits $a_u,b_u,c_u$ and define the two-local gadget Hamiltonian $G$, and the total Hamiltonian $H_{\rm Had}$ as:
\begin{align} 
 \label{eq:Hdiscrete} 
G&=\sum_{u=1}^n \big[-(X_{c_u}+Z_{c_u})-(X_u X_{a_u}+Y_u Y_{a_u}+Z_u Z_{a_u}) \nonumber \\ 
 \nonumber  &-(3X_{a_u}X_{b_u}+Y_{a_u}Y_{b_u}+2Z_{a_u}Z_{b_u})
 -(X_{b_u}X_{c_u}+Y_{b_u}Y_{c_u}+Z_{b_u}Z_{c_u}) \big],\\
 H_{\rm Had} &= H + G.
 \end{align}
Then the following are equivalent:
\begin{enumerate}
\item there exists a unitary $U=\bigotimes_{u=1}^{n} (U_u \otimes U_{a_u} \otimes U_{b_u} \otimes U_{c_u})$ such that $U H_{\rm Had} U^{\dagger}$ is a symmetric $Z$-matrix.
\item there exists $x \in \{0,1\}^n$ such that ${\sf W}(x)^{\dagger} H {\sf W}(x)$ is a symmetric $Z$-matrix, where ${\sf W}(x)=\bigotimes_{u=1}^n W_u^{x_u}$.
\end{enumerate}  
\end{lemma}

\begin{proof}
First we prove $2 \rightarrow 1$. If there exists $x \in \{0,1\}^n$ such that ${\sf W}(x) H {\sf W}(x)^{\dagger}$ is a symmetric $Z$-matrix, then it is easy to check that  $U H_{\rm Had} U^{\dagger}$ is a symmetric $Z$-matrix with
\[U= \bigotimes_{u=1}^{n} \left(W_u \otimes  W_{a_u}\otimes W_{b_u}\otimes W_{c_u}\right)^{x_u}.\]

To prove the other direction, we will show that if 1. holds, each of the single-qubit unitaries $U_{\alpha}$ ($\alpha \in \bigcup_{u=1}^n \{u,a_u,b_u,c_u\}$) must be from the discrete set $\{I,W,X,XW\}$. This fact will suffice by the following reasoning. First note that conjugating by local $X$ matrices permutes the off-diagonal matrix entries of the Hamiltonian among themselves~\cite{marvian2018computational}. So if $U H_{\rm Had} U^{\dagger}$ is a symmetric $Z$-matrix and $U_{\alpha} \in \{I,W,X,XW\}$, then $\bar{U} H_{\rm Had} \bar{U}^{\dagger}$ is also a symmetric $Z$-matrix, where $\bar{U} = \bigotimes_{\alpha} \bar{U}_{\alpha}$ and 
\begin{equation*}
\bar{U}_{\alpha} = \left\lbrace \begin{array}{cc}
I & U_{\alpha} = I \textrm{ or } X\\
W & U_{\alpha}= W \textrm{ or } XW 
\end{array} \right. 
\end{equation*} 
since $UH_{\rm Had} U^{\dagger}$ and $\bar{U} H_{\rm Had} \bar{U}^{\dagger}$ are related by conjugation by local $X$ matrices. Using the fact that the partial trace of a symmetric $Z$-matrix is also a symmetric $Z$-matrix, and noting that by tracing out the ancilla qubits  of $\bar{U} H_{\rm Had} \bar{U}^{\dagger}$ we get $\bar{U} H \bar{U}^{\dagger}$, we conclude that if $\bar{U} H_{\rm Had} \bar{U}^{\dagger}$ is a symmetric $Z$-matrix, then so is $\bar{U}H\bar{U}^{\dagger}$, and $\bar{U} = {\sf W}(x)$ for some $x$.

%
%


We now proceed with proving that $U_{\alpha} \in \{I,W,X,XW\}$ given 1. Here we make use of the picture of orthogonal rotations on $\beta$ matrices, as mentioned in Proposition ~\ref{prop:ortho}. 
For a given $u$ we note that there are no 1-local terms involving qubits $a_u$ and $b_u$, and that the matrix $\beta_{{a_u} {b_u}}$ is diagonal and has 3 distinct non-zero singular values. 
In the absence of 1-local terms, it follows directly from Proposition~\ref{prop:basic} that $\beta_{{a_u} {b_u}}$ has to remain diagonal for $H_{\rm Had}$ to be a symmetric $Z$-matrix.
Therefore, the only possible transformations are signed permutations (of the Paulis) on qubits $a_u$ and $b_u$ with the permutations being the same to maintain the diagonality of $\beta_{{a_u} {b_u}}$.
	This implies that there exists a single-qubit Clifford transformation ${\sf C}$ (corresponding to the permutation) and Pauli matrices $P_{a_u}$ and $P_{b_u}$ such that  $U_{a_u}=P_{a_u}{\sf C}$ and $U_{b_u}=P_{b_u}{\sf C}$.

	We now consider the interaction between qubits $u$ and $a_u$. For the overall Hamiltonian to be real, the coefficients of $X_u Y_{a_u}, Y_u X_{a_u}, Z_u Y_{a_u}, Y_u Z_{a_u}$ must all be zero. Since there are no 1-local terms acting on qubit $a_u$, the coefficient of $Z_u X_{a_u}$ must also be zero and so the rotated matrix $\beta_{u a_u}'$ must have zeroes in the following positions:
	\begin{equation*}\beta_{u a_u}'=O_u^{\sf T} \beta_{u a_u} O_{a_u}=\left(\begin{array}{ccc}
	* & 0 & * \\
	0 & * & 0 \\
	0 & 0 & * \\
	\end{array}\right).
	\end{equation*}
	Note that $-\beta_{u a_u}$ is the identity matrix in Eq.~(\ref{eq:Hdiscrete}), so $\beta_{u a_u}'=O_u^{\sf T} \beta_{u a_u} O_{a_u}=-O_u^{\sf T} O_{a_u}$ is an orthogonal matrix. The only orthogonal matrix with zeroes in these positions is a diagonal matrix (with $\pm 1 $ on the diagonal). Therefore $O_u$ must equal $O_{a_u}$ up to signs; that is $U_u=PU_{a_u}$ for some Pauli $P$.
	
Since the matrix $\beta_{b_u c_u}$ is identical to   $\beta_{u a_u}$ and also there is no 1-local terms acting on qubit $b_u$, an identical argument shows that any curing transformation,  $U_{b_u}$ and $U_{c_u}$ must satisfy $U_{b_u}=P U_{c_u} $ for some Pauli matrix $P$. Thus for all $\alpha \in \{u,a_u,b_u,c_u\}$, we have $U_{\alpha}=P_{\alpha} { \sf C}$ for some Pauli matrix $P_{\alpha}$ and a single-qubit Clifford transformation ${\sf C}$.	
	
	Due to the 1-local terms $-(X_{c_u}+Z_{c_u})$, if ${\sf C}$ maps $X\rightarrow Y$ or $Z \rightarrow Y$ the Hamiltonian will have imaginary matrix entries and so, up to multiplication by a Pauli, ${\sf C}$ must be $I$ or $W$. Incorporating any such Pauli into $P_{c_u}$, we may assume wlog that ${\sf C} \in \{I,W\}$. 
Furthermore, if $P_{c_u}$ is $Y$ or $Z$, there will be a positive $+X_{c_u}$ term, so $P_{c_u} \in \{I,X\}$. Finally, if any of the other $P_{\alpha}$ are $Y$ or $Z$, there will be a positive $+X \otimes X$ term, and so for all $\alpha$ we must have $P_{\alpha} \in \{I,X\}$ and so $U_{\alpha} = P_{\alpha} {\sf C } \in \{I, W,X,XW\}$.
	\end{proof}
	
The following Lemma was proved in \cite{klassen2018two,marvian2018computational} by formulating an efficient strategy which finds a two-local termwise-stoquastic decomposition which is equivalent to $H$ being a symmetric $Z$-matrix by Proposition \ref{prop:global2}\footnote{More generally, one can note that it is easy to decide whether a $k$-local Hamiltonian is $k$-local term-wise stoquastic, as this is a linear programming problem. This can be seen by noting that the number of parameters needed to specify a local decomposition is polynomially dependent on the number of qubits, and the number of conditions to test on each term is dependent on the locality of the term. Furthermore, all of the conditions are linear \cite{marvian2018computational, thesis:ioannou}.}.
\begin{lemma}\label{lem:checkingIsEasy}\cite{klassen2018two}
Given a two-local Hamiltonian $H$ on $n$ qubits, one can decide if $H$ is a symmetric $Z$-matrix in the given basis in a number of steps polynomial in $n$.
\end{lemma}

Corollary~\ref{cor:NPcontainment} now follows immediately from Lemma~\ref{lem:restrictedbasis} and Lemma~\ref{lem:checkingIsEasy}, because the string $(x_1,\dots,x_n)$ is an efficiently checkable witness in the case that $H_{\rm Had}$ is sign-curable by a local unitary transformation:
\begin{corollary}
	\label{cor:NPcontainment}
	If $H$ is a two-local Hamiltonian, then for Hamiltonians $H_{\rm Had}$ of the form in Eq.~(\ref{eq:Hdiscrete}), ${\sf LocalSignCure}$ is in NP.
\end{corollary}

\subsection{{\sf LocalSignCure} is NP-hard}

Now we will show how to reduce 3-SAT to {\sf LocalSignCure}, and hence show that {\sf LocalSignCure} is NP-hard.
At the heart of the construction is a Hamiltonian $H_{\text{OR}}$ which acts on four qubits labeled $d,1,2,3$:
\begin{equation}
H_{\text{OR}}=-(X_d+Z_d+I) \otimes (Z_1+Z_2+Z_3+2I).
\end{equation}
Thanks to Lemma~\ref{lem:restrictedbasis} it suffices to consider a local basis change of the form ${\sf W}(x)=\bigotimes_{j \in \{d,1,2,3\}}W_j^{x_j}$. Note that $-(X_d+Z_d+I)$ has non-positive matrix entries and is invariant under conjugation by $W_d$. Therefore ${\sf W}(x)H_{\text{OR}}{\sf W}(x)^{\dagger}$ is a symmetric $Z$-matrix if and only if the bit string $x$ is such that all the matrix entries of 
\[W_1^{x_1} Z_1 W_1^{x_1} + W_2^{x_2} Z_2 W_2^{x_2} +W_3^{x_3} Z_3 W_3^{x_3} +2I\]
are non-negative. Recalling that $WZW=X$, one can see that for any $x$, all the off-diagonal matrix entries are non-negative. In addition, the diagonal entries are non-negative unless $(x_1,x_2,x_3)=(0,0,0)$. 
Therefore ${\sf W}(x) H_{\text{OR}}{\sf W}(x)^{\dagger}$ is a symmetric $Z$-matrix if and only if $x_1 \vee x_2 \vee x_3$ evaluates to true. 


	Let $C$ be a 3-SAT Boolean formula of the form 
\[C=\bigwedge_{k=1}^m C_k=\bigwedge_{k=1}^m \left(c_{k,1} \vee c_{k,2} \vee c_{k,3}\right),\]
with $m$ clauses and $n$ variables, where each $c_{k,j}$ is equal to $x_i$ or $\bar{x}_i$ for some $i \in \{1,\dots,n\}$.

Let $H_C$ be the Hamiltonian on $m+n$ qubits (labelled $\{1,\dots, n\} \cup \{d_1,\dots d_m\}$) defined by
\begin{equation}
\label{eq:HC}
H_C=\sum_{k=1}^m H_k=\sum_{k=1}^m - (X_{d_k} +Z_{d_k} +I)
\otimes \left(S(c_{k,1})+S(c_{k,2})+S(c_{k,3})+2I\right),
\end{equation}
where  \[S(c)=\left\{ \begin{array}{cc} Z_i & \text{if } c=x_i \text{ for some } i \\
X_i & \text{if } c=\overline{x_i} \text{ for some } i\end{array} \right. .\qquad  \]
An instance of such a Hamiltonian is illustrated in Figure \ref{fig:satConst}. 
\begin{figure}
\includegraphics[scale=1]{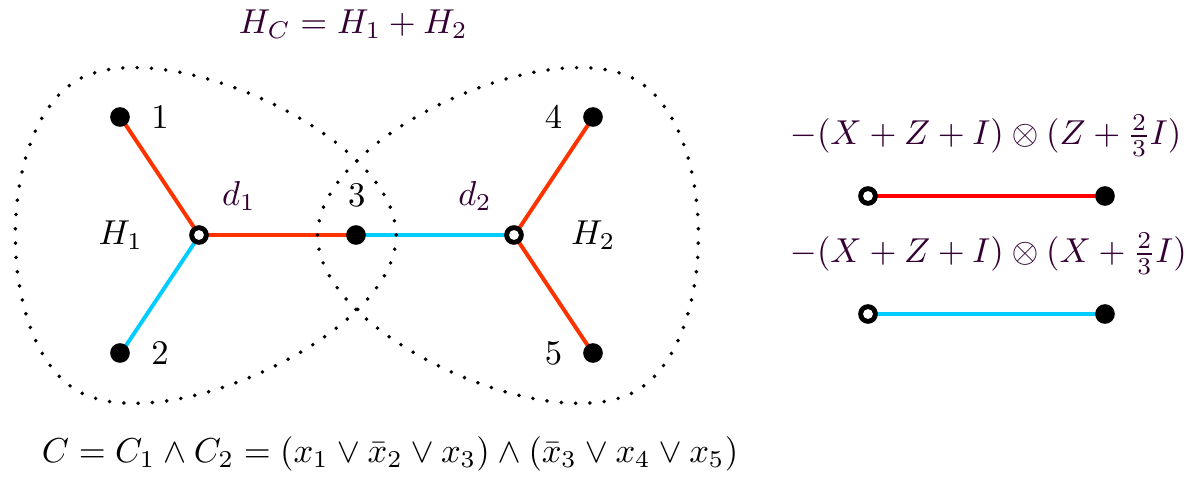}
\caption{An encoding of a $3$-SAT Boolean formula $C$, with two clauses and five variables, into a Hamiltonian $H_C$ as prescribed by Eq.~(\ref{eq:HC}).}\label{fig:satConst}
\end{figure}
For $x \in \{0,1\}^n$ and $y \in \{0,1\}^m$, define 
\[{\sf W}(x,y)=\left(\bigotimes_{i=1}^n W_i^{x_i}\right) \otimes \left( \bigotimes_{j=1}^m W_{d_j}^{y_j}\right).\]

\begin{lemma}
	\label{lem:H3SAT}
	Let $C$ be a 3-SAT Boolean formula, and $H_C$ be the corresponding Hamiltonian defined in Eq.(\ref{eq:HC}), and let $x \in \{0,1\}^n$. $C(x)$ evaluates to true if and only if $\forall y \in \{0,1\}^m$, ${\sf W}(x,y) H_C {\sf W}(x,y)^{\dagger}$ is a symmetric $Z$-matrix.
\end{lemma}
\begin{proof}

Note that $(X_{a_k}+Z_{a_k}+I)$ is invariant under conjugation by $W_{d_k}$, so the choice of $y$ leaves $H_C$ unchanged. Furthermore $(X_{a_k}+Z_{a_k}+I)$ has non-negative matrix entries (with some positive off-diagonal matrix entries). Therefore ${\sf W}(x,y) H_k {\sf W}(x,y)^{\dagger}$ is a symmetric $Z$-matrix if and only if all the matrix entries of 
\begin{equation}
\label{eq:WSSSW}
{\sf W}(x,y) \big(S(c_{k,1})+S(c_{k,2})+S(c_{k,3})+2I\big) {\sf W}(x,y)^{\dagger}
\end{equation}
 are non-negative.
As discussed above, $S(c)$ has been defined so that the matrix entries of (\ref{eq:WSSSW}) are non-negative exactly when $\left(c_{k,1} \vee c_{k,2} \vee c_{k,3}\right)$ is true.

Since each $H_k$ is the only interaction acting on qubit $d_k$, and $H_k$ can only fail to be a symmetric $Z$-matrix due to terms which act non-trivially on $d_k$, it follows that \linebreak ${\sf W}(x,y) H_k {\sf W}(x,y)^{\dagger}$ must be a symmetric $Z$-matrix for all $k$, in order for \linebreak ${\sf W}(x,y) H_C {\sf W}(x,y)^{\dagger}$ to be a symmetric $Z$-matrix. Since $C=\bigwedge_{k=1}^m C_k$, this happens exactly when $C(x)$ is true.
\end{proof}

This leads to the main result of this section:
\begin{corollary}
There exists a class of two-local Hamiltonians for which {\sf LocalSignCure} is NP-complete.
\end{corollary}

\begin{proof}
For any 3-SAT formula $C$, we construct the two-local Hamiltonian $H_{C,{\rm Had}}$ by adding the gadget interactions $G$ of Eq. (\ref{eq:Hdiscrete}) for each qubit in the Hamiltonian $H_C$ in Eq.~(\ref{eq:HC}). Using  Lemma~\ref{lem:restrictedbasis} and Lemma~\ref{lem:H3SAT}, 
we conclude that satisfying a family of 3-SAT formulae $C$ is equivalent to {\sf LocalSignCure} for the corresponding family of $H_{C,{\rm Had}}$ Hamiltonians, from which we conclude that {\sf LocalSignCure} is NP-hard.
The inclusion of {\sf LocalSignCure} for $H_{C,{\rm Had}}$ in NP follows from  Corollary~\ref{cor:NPcontainment}.
\end{proof}


Let us briefly comment on the question how hard determining the groundstate energy of $H_{C, {\rm Had}}$ may be. We observe that the qubits $d_i$ and the triples of ancilla qubits $a_u,b_u,c_u$ for each $u$ only couple to the $n$ qubits on which the clauses act. 
In particular, if we would fix the state of these ancillary qubits to $\psi$, then the resulting Hamiltonian $\bra{\psi} H_{C, {\rm Had}} \ket{\psi}$ acting only on the clause qubits would be purely 1-local. It is however not a priori clear that the minimal energy is obtained when the state $\psi$ is a product state, if this were the case then the ground state energy problem would be in NP as a prover could provide a description of this product state. However even the problem of finding such a product ground state is not guaranteed to have an obvious polynomial-time classical algorithm.
It would be worthwhile to investigate this further.



\section{An efficient algorithm for {\sf LocalSignCure} for exactly two-local Hamiltonians}\label{sec:curingResult}


\subsection{Preliminaries}\label{sec:prelim}

In this section we prove Theorem \ref{thm2} by presenting an efficient algorithm for solving {\sf LocalSignCure} when $H$ is an exactly two-local Hamiltonian.

We represent an exactly two-local Hamiltonian by a graph $G$ with matrix-weighted edges. Each qubit in the Hamiltonian corresponds to a vertex in the graph, and each edge corresponds to a term $H_{uv} \neq 0$. Every edge is weighted by the $3\times 3$ real matrix $\beta_{uv}$ associated with $H_{uv}$, as discussed in Proposition~\ref{prop:ortho}. 

 In this picture,  {\sf LocalSignCure} reduces to the following problem.  Consider a graph $G=(V,E)$ with $n$ vertices in $V$ and a set of directed matrix-weighted edges $E$. Each edge $(u,v)$ with direction $u \rightarrow v$ is weighted by a $3\times 3$ real matrix $\beta_{uv}$, and we define $\beta_{vu} = \beta_{uv}^{\sf T}$~\footnote{The purpose of the direction is merely to allow the matrix weight to be well defined. Throughout the text we will ignore the directedness of the graph, and treat the edge as though it is weighted by $\beta_{uv}$ or $\beta_{vu}$ depending on our purpose.} 
Given $G$, find a set of $SO(3)$ rotations $\{O_u\}_{u=1}^n$ which have the action $O_u^{\sf T} \beta_{uv} O_v = \Sigma_{uv} \;  \forall \beta_{uv}$, such that for all edges $(u,v)$:
\begin{equation} \label{cond:diagonality0}
\Sigma_{uv} \textrm{ is a diagonal matrix},
\end{equation}
\begin{equation} \label{cond:magnitude0}
\vert (\Sigma_{uv})_{11} \vert \geq \vert (\Sigma_{uv})_{22} \vert \; \;\;\forall \beta_{uv},
\end{equation} 
\begin{equation} \label{cond:negativity0}
(\Sigma_{uv})_{11} \leq 0 \; \;\;\forall \beta_{uv}.
\end{equation}
Otherwise prove that no such set exists. 

Note that we have rephrased the conditions in Proposition~\ref{prop:basic} according to the labeling $X \rightarrow 1$, $Y \rightarrow 2$, $Z\rightarrow 3$. One can argue, see Ref.~\cite{klassen2018two}, that if there exist $O(3)$ rotations that perform this task, then one can easily construct a set of $SO(3)$ rotations that do the same. Therefore any orthogonal rotations will suffice. 

If all matrices $\beta_{uv}$ are diagonal, then the XYZ-algorithm in Theorem~\ref{theo:KT} can be applied. Naively, our problem could then be reduced to the question: Is there a set of rotations $\{O_u\}$ that has the action $O_u^{\sf T} \beta_{uv} O_v = \Sigma_{uv} \;  \forall \beta_{uv}$, such that condition~\ref{cond:diagonality0} is satisfied, and what are those rotations? If this problem is efficiently solved, one may incorporate the algorithm for finding the set of rotations as a sub-routine of the XYZ-algorithm and solve the entire problem. However we show in Appendix~\ref{app:HeisForm} that deciding the existence of such a set of rotations on $\beta_{uv}$ such that condition~\ref{cond:diagonality0} is satisfied is an NP-hard problem.  


Thus a different approach must be taken, namely we focus on condition~\ref{cond:magnitude0} in order to prune the set of solutions which needs to be considered. More concretely, we will present an algorithm which solves the following problem:

\begin{problem*}[{\sf No-Lone-YY $\&$ Diagonal}]\label{problemStatement}
Is there a set of orthogonal rotations $\{O_u \in O(3)\}$ that have the action $O_u^{\sf T} \beta_{uv} O_v = \Sigma_{uv} \;  \forall \beta_{uv}$, such that:
\begin{equation}\label{cond:diagonality}
\Sigma_{uv} \textrm{ is a diagonal matrix}, 
\end{equation}
\begin{equation} \label{cond:magnitude}
 (\Sigma_{uv})_{22} = 0, \; \;\;\forall \beta_{uv} \textrm{ for which } \textrm{Rank}\,(\beta_{uv}) =1.
\end{equation} 
If yes, what is that set? Note here that condition~\ref{cond:diagonality} is identical to condition~\ref{cond:diagonality0}, and condition~\ref{cond:magnitude} is precisely condition~\ref{cond:magnitude0} restricted to rank-1 matrices. 
\end{problem*}

Note that an efficient algorithm for this problem can be incorporated into the XYZ-algorithm to produce an efficient algorithm for ${\sf LocalSignCure}$ for exactly two-local Hamiltonians, thus directly proving Theorem~\ref{thm2}. More precisely, a solution to this problem prescribes a transformation of our Hamiltonian into an XYZ-Heisenberg Hamiltonian, in which case the XYZ-algorithm can be used to decide if the Hamiltonian can be rotated into a symmetric $Z$-matrix by single-qubit unitary transformations. Furthermore, if no solution exists to this problem, then rotating the Hamiltonian into a symmetric $Z$-matrix by single-qubit transformations is impossible, since both of the above conditions are necessary conditions.

An orthogonal transformation $O_u$ can be written as $O_u=(e_u^1, e^2_u, e_3^u)$ where the $e_u^i$ are three real orthonormal column vectors. We can thus view selecting $O_u$ as selecting a basis $b_u=(e_u^1, e_u^2, e_u^3)$ at vertex $u$.

\begin{definition}[No-Lone-YY Basis (NLY Basis)]
Given a matrix-weighted \linebreak graph $G$ with weights $\beta_{uv}$, an ordered assignment of basis vectors $b_u=(e_1^u, e_2^u, e_3^u)$ to each vertex in the graph is called a \textbf{No-Lone-YY basis (NLY basis)} $B=\{b_u\}$, when $e_i^v$ is a right singular vector of $\beta_{uv}$ with corresponding left singular vector equal to $\pm e_i^u$, i.e. 
\begin{equation} \label{eq:NLYcond1}\forall u,v, i: \; \beta_{uv} e_i^v=\pm \sigma e_i^u ,\;\; \beta_{uv}^{\sf T} e_i^u=\pm \sigma e_i^v,
\end{equation}  and for all rank $1$ matrices $\beta_{uv}$:   
\begin{equation}\label{eq:NLYcond2} \beta_{uv} e_2^v=0 ,\;\; \beta_{uv}^{\sf T} e_2^u=0.
\end{equation}
\end{definition}

It is not hard to see that solving problem statement~\ref{problemStatement} is equivalent to finding a NLY basis, or showing that none exists. 

It is important to note that if we flip the signs on our basis elements, this will have no bearing on the problem. We formally define this equivalence under sign flips as 
\begin{definition}
Two ordered bases $b_u$ and $b_u'$ are  equivalent modulo signs:
$$b_u = b_u' \text{ modulo signs} $$
if $b_u = (e_1^u, e_2^u, e_3^u)$, and $b_u' = (\delta_1^u e_1^u, \delta_2^u e_2^u, \delta_3^u e_3^u)$  with $\delta_i^u \in \{+1, -1\}$.  
\end{definition}
Thus throughout the text we will often talk about a basis \textbf{modulo signs}, meaning a basis choice where the signs have not been specified. The premise is that the choice of signs is irrelevant for the purposes of the problem. This will prove to be a useful fact in the proofs of Lemma~\ref{lem:genericAnsatz} and Theorem~\ref{thrm:genericAlgo}.




A final comment on notation. In the next two subsections we will make use of \textbf{sets of subspaces} of $\mathbb{R}^3$. We wish to hold onto the notion that these are sets of subspaces, but make use of natural set notation in terms of the elements of the subspaces. Consequently, for ease of exposition, we will abuse notation in the following ways. We denote a set of subspaces by $\mathbb{S}=\{S_i \vert S_i \subseteq \mathbb{R}^3\}$. We denote the entrywise intersection of sets of subspaces by 
\begin{equation*}
 \mathbb{S}_1 \cap \mathbb{S}_2 :=  \{S_i \cap S_j \vert S_i \in \mathbb{S}_1 \;, S_j \in \mathbb{S}_2\}.
 \end{equation*}
We denote the span of the union of the subspaces by 
\begin{equation*}
 \textrm{span} (\mathbb{S}) := \textrm{span}\left( \bigcup_{S_i \in \mathbb{S}} S_i \right).
\end{equation*}
We say a set of vectors $b= \{\nu \vert \nu \in \mathbb{R}^3 \}$ is in a set of subspaces $\mathbb{S}$, with the notation $b \subseteq \mathbb{S}$, if every vector in $b$ belongs to a subspace in $\mathbb{S}$. Furthermore, we say a set of subspaces $\mathbb{S}_1$ is contained in another set of subspaces $\mathbb{S}_2$, with the notation $\mathbb{S}_1 \subseteq \mathbb{S}_2$, if every subspace in $\mathbb{S}_1$ is contained in a subspace in $\mathbb{S}_2$. The reason these two notations coincide is because it can be helpful for our purposes to conceptualize the vectors in $b$ as 1-dimensional subspaces, since we do not care about the sign of the vector.
We denote the transformation on each of the subspaces by an orthogonal rotation $O$ as:
\begin{equation*}
O\mathbb{S} := \{ OS_i \vert S_i \in \mathbb{S} \}.
\end{equation*}

\subsection{XOR-SAT}
In the next two subsections we will make repeated use of a subroutine for solving the 2-XOR-SAT problem. XOR-SAT is a Boolean satisfiability problem in which one has a set of Boolean variables $\{x_u\}$ and a set of clauses consisting of not operations and xor operations e.g. $ \bar{x}_u \oplus x_v$, and one asks if there exists an assignment to the Boolean variables which satisfies all of the clauses. XOR-SAT is known to be solvable in polynomial time. 2-XOR-SAT is quite trivially solvable in time $O(N^2)$, where $N$ is the number of variables: the assignment of one variable in the clause uniquely determines the assignment of the other variable in the clause. Thus one varies the assignment of one variable, and propagates that choice through the clauses (of which there are worst case $N^2$), until all variables are assigned or a contradiction is found (if there are disconnected sets of variables, one does the same thing for each disconnected cluster).


\subsection{Illustrative sub-case: graphs with rank-1 edges}\label{subsec:rank1graph}
We begin by considering an illustrative sub-case, where each edge in the graph is weighted by a \lowrank\ matrix (i.e. a \textbf{\lowrank\ edge}).
The significance of \lowrank\ edges is that their matrix weights have a two dimensional null space, which implies an additional freedom in the choice of basis that is not present in edges weighted by \highrank\ matrices (i.e \textbf{\highrank\ edges}), which have at most a one-dimensional null space. This difference will become more apparent when we consider the general case of a graph with both \highrank\ and \lowrank\ edges. 


For a graph with only \lowrank\ edges the algorithm for solving problem statement~\ref{problemStatement}  breaks up into two parts. In the first part we impose some of the necessary constraints for the basis assignment to be NLY, formulating a candidate basis $B$. 
In the second part we permute the vectors of the candidate basis so that it could become an NLY basis.


\begin{definition}[Candidate Basis of a \lowrank\ graph] \label{def:ansatzRank1}
A \textbf{candidate basis of a \lowrank\ graph} is a basis assignment $B= \{b_u\}$ such that for every edge $e=(u,v)$, the basis vectors $b_u = (e_1^u, e_2^u, e_3^u)$ are eigenvectors of $\beta_{uv} \beta_{uv}^{\sf T}$ and the basis vectors $b_v= (e_1^v, e_2^v, e_3^v)$ are eigenvectors of $\beta_{uv}^{\sf T} \beta_{uv}$.
\end{definition}

\begin{proposition} \label{prop:lowrank}
Given a \lowrank\ matrix $\beta_{uv}$, if the basis vectors $b_u $ are eigenvectors of $\beta_{uv} \beta_{uv}^{\sf T}$ and the basis vectors $b_v$ are eigenvectors of $\beta_{uv}^{\sf T} \beta_{uv}$, then there exists a single index $i$  such that $\beta_{uv}^{\sf T} e_i^u \neq 0$  and a single index $j$ such that $\beta_{uv} e_j^v \neq 0$. Furthermore, $\exists \sigma \neq0$ s.t. $\beta_{uv}e_i^v = \pm \sigma e_j^u$ and $\beta_{uv}^{\sf T}e_j^u = \pm \sigma e_i^v$.  
\end{proposition}


\begin{proof}
Since $\beta_{uv}$ is \lowrank\ it follows that $\beta_{uv} \beta_{uv}^{\sf T}$ and $\beta_{uv}^{\sf T} \beta_{uv}$ are also \lowrank . Thus only single basis vectors $e_i^u \in b_u$ and $e_j^v \in b_v$ will be eigenvectors with non-zero eigenvalue of $\beta_{uv} \beta_{uv}^{\sf T}$ and $\beta_{uv}^{\sf T} \beta_{uv}$ respectively. Therefore $e_i^u$ and $e_j^v$ are the only singular vectors in $b_u$ and $b_v$ which have non-zero singular values for $\beta_{uv}$. Since the column and row spaces of $\beta_{uv}$ are both 1-dimensional, it must be the case that  $\beta_{uv}e_i^v = \pm \sigma e_j^u$ and $\beta_{uv}^{\sf T}e_j^u = \pm \sigma e_i^v$ for some $\sigma \neq 0$.
\end{proof}

Note that given a candidate basis (and the corresponding orthogonal rotations $\{O_u\}$) the matrix $O_u^{\sf T} \beta_{uv} O_v$ has exactly one non-zero entry but isn't necessarily diagonal. An example of a matrix of this form would be:
\begin{equation}\label{eq:lowRankExample}
O_u^{\sf T} \beta_{uv} O_v = \begin{bmatrix}
0&0&4\\
0&0&0\\
0&0&0
\end{bmatrix}.
\end{equation} 

Therefore a candidate basis is close to being a NLY basis, except the ordering of the basis vectors in $b_u$ and $b_v$ may not be correct. In order to remedy this, we need to permute orderings of the various $b_u$. 
To help visualize this, we may consider the edge $(u,v)$ to be labelled by $i$ on the $u$ side, and $j$ on the $v$ side, where $i$ and $j$ are the indices specified in Proposition~\ref{prop:lowrank}. For example, the matrix in Eq.~(\ref{eq:lowRankExample}) would correspond to the edge in Figure~\ref{fig:bilabeledEdge}.
\begin{figure}[htb]
\begin{center}
\includegraphics[scale=.7]{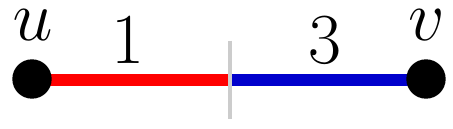}
\end{center}
\caption{Bilabelling of a \lowrank\ edge}\label{fig:bilabeledEdge}
\end{figure}
In this picture the candidate basis $B$ thus specifies a bi-labelled graph, i.e. a graph where every edge has two labels (two colors).

\begin{definition}[Basis Permutation]
Given a basis $b = (e_1, e_2, e_3)$ and permutation $\pi$, the \textbf{permuted basis $b_u^{\pi}$} is defined as:
$b_u^{\pi} := (e_{\pi^{-1}(1)}^u ,e_{\pi^{-1}(2)}^u ,e_{\pi^{-1}(3)}^u  ).$ Given a basis assignment to every vertex $B = \{b_u\}$ and an assignment of permutations to every vertex $\Pi = \{\pi_u\}$, the permuted basis assignment is defined as $B^{\Pi}: = \{ b_u^{\pi_u} \}$

\end{definition}

 Given that the candidate basis $B$ specifies a bi-labelled graph, we can think of the action of basis permutations 
$ b_u \rightarrow b_u^{\pi} $ as a transformation on the labeling, $i \rightarrow \pi(i)$, of every label adjacent to $u$, as illustrated in Figure~\ref{fig:bilabeledGraph}.
 \begin{figure}[htb]
\begin{center}
\includegraphics[scale=.7]{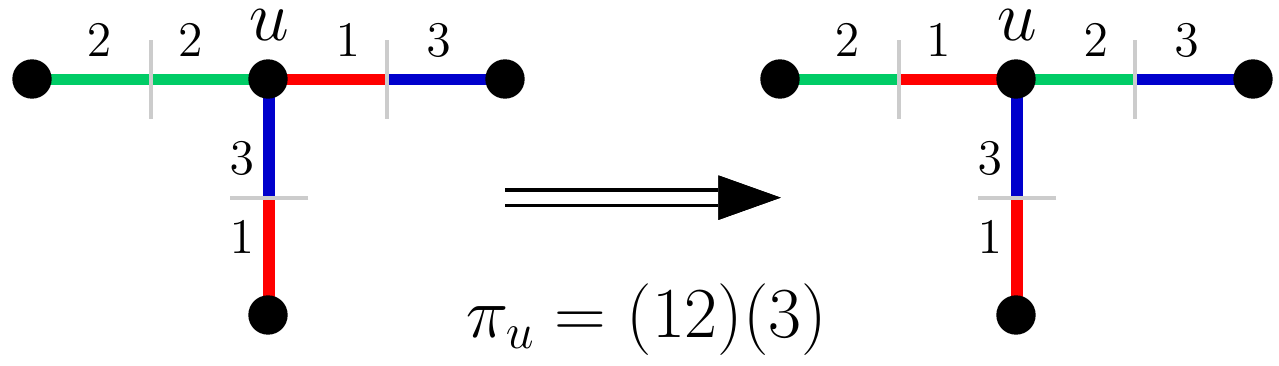}
\end{center}
\caption{Action of permutations on a bi-labelled graph}\label{fig:bilabeledGraph}
\end{figure} 
 The premise is then that the only remaining task is to find a set of permutations $\{\pi_u \in S_3\}$ to apply to every vertex so that:
\begin{itemize}
\item  The bi-labelling is uniform on an individual edge (i.e. $i=j$), corresponding to condition~\ref{eq:NLYcond1}.
\item no edge is labelled by the (green) value $2$, corresponding to condition~\ref{eq:NLYcond2}.
\end{itemize}

If we are unsuccessful in either finding a candidate basis $B$, or an appropriate permutation $\Pi$, then we will argue that no NLY basis exists.

\begin{algorithm}
    \SetKwInOut{Input}{Input}
    \SetKwInOut{Output}{Output}

    \Input{Graph $G = (V, E)$, \lowrank\ matrix edge weights $\{\beta_{uv}\}$ }
    \Output{A candidate basis $B= \{b_u\}$, if one exists. Otherwise False, indicating no candidate basis exists.}
    \For{$v \in V$}{
    	$\mathbb{S}_v = \{ \mathbb{R}^3 \}$
    	
    	\For{$u \in V$ s.t. $e=(u,v) \in E$}{
    		$\mathbb{S}_v^e= $  the set of orthogonal maximal eigenspaces associated with every eigenvalue of $\beta_{uv}^{\sf T} \beta_{uv}$
    		
    		$\mathbb{S}_v = \mathbb{S}_v \cap \mathbb{S}_v^e$
    	}
    	
    	\If{$\textrm{span}(\mathbb{S}_v) \neq \mathbb{R}^3$ }
      {
        return False\;
      }
      Choose orthonormal basis $b_v = (e_1^v, e_2^v, e_3^v) \subseteq \mathbb{S}_v$
    }
      return $B= \{b_v \}$
 
    \caption{Algorithm for finding a candidate Basis of a \lowrank\ graph}\label{alg:rank1ansatz}
\end{algorithm}

\begin{lemma}\label{lem:rank1asatz}
Algorithm~\ref{alg:rank1ansatz} efficiently finds a candidate basis for a \lowrank\ graph, or otherwise shows that no such candidate basis exists.
\end{lemma}
\begin{proof}
Any vectors we choose from  $\mathbb{S}_v$ must simultaneously be eigenvectors of $\beta^{\sf T}_{uv} \beta_{uv}$ for all edges $e=(u,v)$  adjacent to $v$, since they must simultaneously belong to every $\mathbb{S}_v^e$. Furthermore, the spaces in $\mathbb{S}_v$ contain all vectors that are simultaneously eigenvectors of $\beta^{\sf T}_{uv} \beta_{uv}$ for all edges $e=(u,v)$  adjacent to $v$. Therefore if $\mathbb{S}_v$ does not span $\mathbb{R}^3$, then we cannot possibly choose a set of orthonormal vectors $b_v$ which are simultaneous eigenvectors of all neighbouring edges. 

The number of elements in any set of eigenspaces $\mathbb{S}_v^e$ is upper bounded by 3, corresponding to 3 orthogonal 1-dimensional subspaces. The same is true for any intersection of any number of these sets. Thus computing any intersection between these sets of subspaces takes $O(1)$ time. Therefore one may iteratively construct $\mathbb{S}_v$ in time proportional to the number of edges. Thus the algorithm is efficient.
\end{proof}

\begin{algo}[Finding permutations $\Pi$ such that $B^{\Pi}$ is an NLY basis] 
This algorithm takes a candidate basis $B$ of a \lowrank\ graph and  finds a set of permutations $\Pi$ such that $B^{\Pi}$ is a NLY basis, or otherwise indicates that no such set of permutations exist.

 For each edge $(u,v)$, identify the left singular vector $e_i^u \in b_u$ and corresponding right singular vector $e_j^v \in b_v$ which are not in the null space of $\beta_{uv}$, which must exist by Proposition~\ref{prop:lowrank}.  Label each \lowrank\ edge $(u,v)$ with an ordered pair of labels $(i,j)$, as illustrated in figure~\ref{fig:bilabeledEdge}. We say that an edge $e=(u,v)$ with labelling $(i,j)$ connects to $u$ with label $i$ and connects to $v$ with $j$.

If for any vertex $v$ there are at least three edges, each connected to $v$ by a different label, then terminate and indicate that the desired set of permutations does not exist.

If the algorithm has not terminated,  then for every vertex $v$ there exist two labels $i$ and $j$ such that every edge adjacent to $v$ connects to $v$ with one of those two labels. This holds even if every edge connects to $v$ with the same label. Identify a pair of permutations $\pi^0_{v}$ and $\pi^1_{v}$ such that $\pi^0_{v}(i)=1$, $\pi^0_{v}(j)=3$ and $\pi^1_{v}(i)=3$, $\pi^1_{v}(j)=1$.

The task now becomes assigning a binary value $x_{v}$ to each $v$ so that for every edge $(u,v)$ with label $(i,j)$ the binary assignments satisfy 
\begin{equation*}
\pi^{x_{u}}_{u}(i) = \pi^{x_{v}}_{v}(j).
\end{equation*}
 By virtue of $\pi$ never mapping any label to the value $2$, and ensuring the uniform bi-labelling of each edge, such an assignment will specify an NLY basis. 

This problem reduces straightforwardly to an XOR-SAT problem. Each edge $(u,v)$ corresponds to an XOR clause: $\bar{x}_{u} \oplus x_{v}$ if $i=j$, and $ x_{u} \oplus x_{v}$ if $i \neq j$. If there is a solution, then this specifies an NLY basis, namely $B^{\Pi}$ with $\Pi= \{\pi^{x_{u}}_{u}\}$. If there is no solution, then the desired set of permutations does not exist.
\end{algo}

\begin{algorithm} 
    \SetKwInOut{Input}{Input}
    \SetKwInOut{Output}{Output}

    \Input{Graph $G = (V, E)$, \lowrank\ matrix edge weights $\{\beta_{uv}\}$, candidate basis $B = \{b_u\}$ }
    \Output{A set of permutations $\Pi = \{ \pi_u\} $ such that $B^{\Pi}$ is an NLY basis, if one exists. Otherwise False.}
 	\For{$v \in V$}{
 		\tcc{Label all incident edges according to which basis vector in $b_v$ is not in the null space of $\beta_{uv}$. These always exist by Proposition~\ref{prop:lowrank}.}
 		$L(v) =\{ \}$
 		
 		\For{$u\in V$ s.t. $e=(u,v)\in E$}{
 			\For{$i\in \{1,2,3\}$} {
			$e_i^v = b_v[i]$
			
			\If{$\beta_{uv}e_i^v\neq0$}							{	
				$L(v,e) = i$
				
				$L(v) = L(v) \cap \{i\}$
			}
			}
		}
		\tcc{If a vertex is incident on more than two different labels, then return false.}
			\If{$\vert L(v) \vert = 3$}{
				return False
			} 	
			\tcc{If the algorithm has not terminated,  then for every vertex $v$ there exist at most two labels such that every edge incident on $v$ connects to $v$ with one of those two labels.}			
			
			\tcc{Define permutations so that all incident edge labels are mapped to either $1$ or $3$.}
			Choose permutation $\pi_v^0$ s.t. $\pi_v^0(L(v)[1])=1$, and if $\vert L(v)\vert>1$ then $\pi_v^0(L(v)[2])=3$
			
			Choose permutation $\pi_v^1$ s.t. $\pi_v^1(L(v)[1])=3$, and if $\vert L(v)\vert>1$ then $\pi_v^1(L(v)[2])=1$ 
 		}
 		
		\tcc{For each edge define a 2-XOR-SAT clause.} 		
 		
 		\For{$e=(u,v)\in E$}{
 			\eIf{$L(v,e) = L(u,e)$}{
 				$C_e(x_u,x_v) = x_u \oplus x_v$
 			}{
				 $C_e(x_u,x_v) = \bar{x}_u \oplus x_v$			
 			}
 		}
 		\tcc{The solution to the associated 2-XOR-SAT problem specifies which permutations to apply at each vertex so that $\pi_u^{x_u}(i) = \pi_v^{x_v}(j)$ . By virtue of $\pi$ never mapping any label to the value $2$, and ensuring the uniform bi-labelling of each edge, such an assignment will specify an NLY basis.} 
 		success= 2-XOR-SAT(ref $\{ x_v \;\vert\; \forall v \in V\}$, $\{C_e \;\vert\; \forall e \in E\}$)

 		\eIf{success}{
 			return $\Pi = \{ \pi_v^{x_v}\;\vert\; \forall v \in V\}$
 		}{
 			return False
 		}

    \caption{Algorithm for finding permutations $\Pi$ such that $B^{\Pi}$ is an NLY basis}\label{alg:rank1perms}
\end{algorithm}



\begin{theorem}\label{thrm:rank1algo}
Given a graph with only \lowrank\ edges, one can efficiently find an NLY basis, or otherwise show that no such basis exists.
\end{theorem}

\begin{proof}
The algorithm for finding an NLY basis in this case proceeds by first finding a candidate basis $B$ using Algorithm~\ref{alg:rank1ansatz}, and then finding a set of permutations $\Pi$ such that $B^{\Pi}$ is a NLY basis using Algorithm~\ref{alg:rank1perms}.   It should be clear that the basis $B^{\Pi}$ is an NLY basis, since for every edge $(u,v)$ Algorithm~\ref{alg:rank1perms}  has explicitly paired those two vectors in $b_u$ and $b_v$ not in the null space of $\beta_{uv}$, and ensured that they are not the second entry. Additionally, Algorithm~\ref{alg:rank1perms} is efficient, since solving 2-XOR-SAT is efficient.

If Algorithm~\ref{alg:rank1ansatz} fails, then by Lemma~\ref{lem:rank1asatz} no candidate basis exists, and since any NLY basis must satisfy the conditions of being a candidate basis, no NLY basis exists.  Furthermore, when given a candidate basis $B$, if Algorithm~\ref{alg:rank1perms} fails, then clearly no set of permutations $\Pi$ exists such that $B^{\Pi}$ is an NLY basis. In one case this is because there are three edges connected to a vertex by a different label, and thus the label $2$ cannot be removed by any permutation. In the other it is because there is not a solution to the 2-XOR-SAT problem, which rules out all potential permutions for those vertices connected to exactly two labels, while in the case of vertices connected to exactly one label, there are other possible permutations, but they would have the same action, and are thus also ruled out.

%


The only non-trivial fact left to prove is that if, given a candidate basis $B$, Algorithm~\ref{alg:rank1perms} fails, then no NLY basis exists. Naively once could imagine that, given some alternative candidate basis, Algorithm~\ref{alg:rank1perms} might succeed. Here we prove that this cannot happen, using proof by contradiction. 

Assume that given a candidate basis $B$, Algorithm~\ref{alg:rank1perms} fails and there does not exist a permutation $\Pi$ such that the basis $B^{\Pi}$ is an NLY basis. Suppose however that there exists an NLY basis $\bar{B}$. If for some edge $(u,v)$ adjacent to $u$, the basis vector $e_i^u \in b_u$ satisfies $\beta_{uv}^{\sf T} e_i^u \neq 0$, then there must exist a unique vector $\bar{e}_j^u \in \bar{b}_u$ such that $\beta_{uv}^{\sf T} \bar{e}_j^u \neq 0$. Furthermore $e_i^u = \pm \bar{e}_j^u $, since $\beta_{uv}$ is \lowrank . Therefore for every edge $(u,v)$ adjacent to $u$, if $e_i^u$ satisfies $\beta_{uv}^{\sf T} e_i^u \neq 0$, then $\bar{e}_j^u$ also satisfies $\beta_{uv}^{\sf T} \bar{e}_j^u \neq 0$. Therefore for every index $i \in \{1,2,3\}$ there exists an index $j \in \{1,2,3\}$ such that for every edge $e=(u,v)$ adjacent to $u$, if $e_i^u $ satisfies $\beta_{uv}^{\sf T} e_i^u \neq 0$ then $\bar{e}_j^u$ satisfies $\beta_{uv}^{\sf T} \bar{e}_j^u \neq 0$. Let $\pi_u$ be the permutation with the mapping: $\pi_u(i)=j$, and $\Pi =\{\pi_u\}$. Then the bi-labelled graph associated with $B^{\Pi}$ must be identical to the bi-labelled graph associated with $\bar{B}$, and therefore $B^{\Pi}$ must be an NLY basis, which is a contradiction.
\end{proof}

\subsection{Graphs with both \highrank\ and \lowrank\ edges}
We will now show how the intuition and arguments given in Subsection~\ref{subsec:rank1graph} translate into the case where the matrix weights may have any rank. First we outline the structure of the argument. Just as in Subsection~\ref{subsec:rank1graph}, we will first search for a candidate basis $B$ for the graph, and then search for an appropriate set of permutations $\Pi$ to apply to the basis vectors. 


\begin{definition}[Candidate Basis of a Graph]\label{def:ansatz}
 A \textbf{candidate basis} $B=\{b_u\}$ is an assignment of basis vectors $b_u= (e_1^u, e_2^u, e_3^u)$, to each vertex $u$, satisfying the following two conditions:
\begin{enumerate}
\item For every \lowrank\ edge $(u,v)$ adjacent to the vertex $u$, the basis vectors $b_u$ are eigenvectors of $\beta_{uv} \beta_{uv}^{\sf T}$.
\item For every \highrank\ edge $(u,v)$, the basis vectors of $b_u$ and $b_v$ are left and right singular vectors of $\beta_{uv}$ respectively, and satisfy:
$\exists  \sigma \in \mathbb{R}$ s.t. $ \beta_{uv} e_i^v = \pm \sigma e_i^u  \textrm{ and }\; \beta_{uv}^{\sf T} e_i^u = \pm \sigma e_i^v$.
\end{enumerate}


\end{definition}
The candidate basis has the same requirements on \lowrank\ edges as in the previous section, however it satisfies more stringent requirements on \highrank\ edges, namely that the transformed matrix weights $O_u^{\sf T} \beta_{uv} O_v$ are diagonal under the prescribed orthogonal rotations $\{ O_u\}$. The most significant difference between the algorithms presented in this section 
is the procedure for finding a candidate basis (Algorithm~\ref{alg:genericAnsatz}). However once a candidate basis has been found, the procedure for finding an appropriate set of permutations (Algorithm~\ref{alg:genericPerms}) will have the same essential form as Algorithm~\ref{alg:rank1perms} with one difference: Instead of individual vertices being the sites to which permutations are assigned, we will instead assign permutations to subgraphs whose vertices are connected by \highrank\ paths (Definition~\ref{def:HRCC}), so that each vertex in such a subgraph is permuted uniformly. This is illustrated in Figure~\ref{fig:permGenericGraph}, in contrast to Figure~\ref{fig:bilabeledGraph}. It will be straightforward to see that if Algorithms \ref{alg:genericAnsatz} and \ref{alg:genericPerms} succeed, then they will have produced an NLY basis. The only significant subtle point that remains, and will be argued in Theorem~\ref{thrm:genericAlgo}, is that if Algorithm~\ref{alg:genericPerms} is given a candidate basis and fails to find a set of permutations which produces an NLY basis, then no NLY basis exists and in particular no other candidate bases need to be considered.

\begin{figure}\label{fig:permGenericGraph}
\includegraphics[scale=.7]{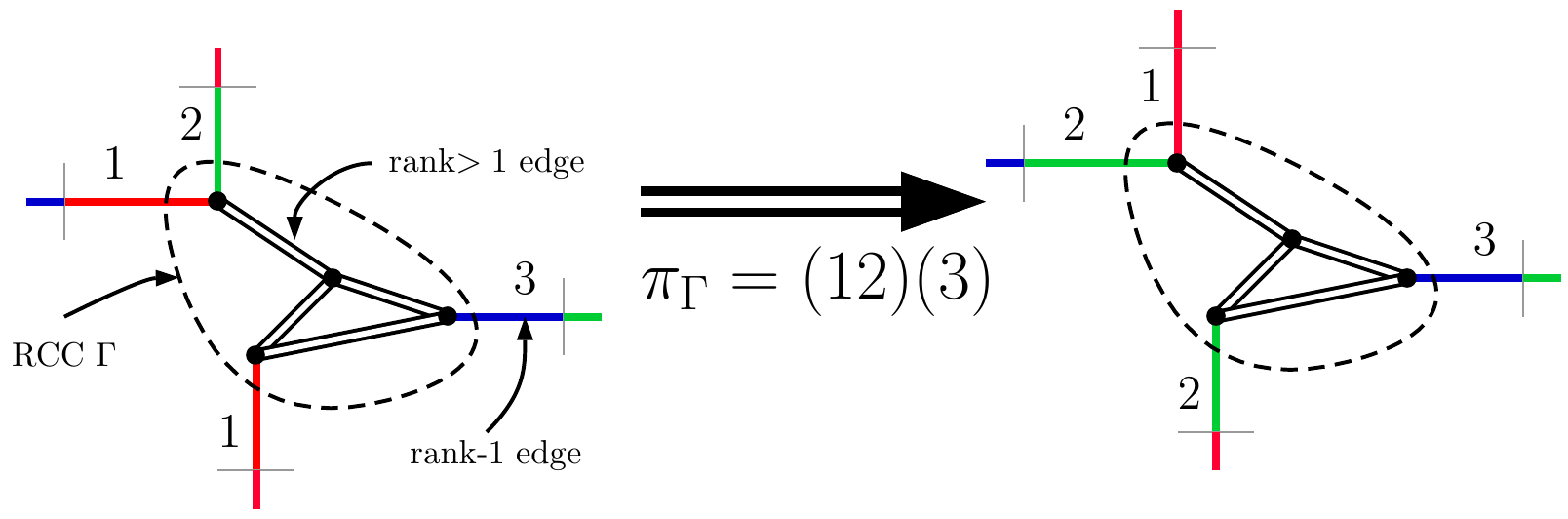}
\caption{Action of a permutation $(12)(3)$ on a \HRCC \;(RCC) in black.}
\end{figure}

Before proceeding with the description of the algorithm for finding a candidate basis, we must establish some facts about \highrank\ edges, and the structure they impose on the problem.

\begin{lemma}\label{lem:rigid} 
Given a \highrank\ edge $(u,v)$, and bases $b_u$, $b_v$ which are eigenvectors of $\beta_{uv}\beta_{uv}^{\sf T}$ and $\beta_{uv}^{\sf T} \beta_{uv}$ respectively. The vectors $e_i^u \in b_u$ and $e_i^v \in b_v$ satisfy $\beta_{uv} e_i^v = \pm \sigma_i e_i^u$ and $\beta_{uv}^{\sf T} e_i^u = \pm \sigma_i e_i^v$, $\sigma_i \in \mathbb{R}$, if and only if, for every singular value decomposition $\beta_{uv} = O_u^e \Sigma^{\rm  SVD}_{uv} (O_v^e)^{\sf T}$ the operator defined as 
\begin{equation} \label{eq:transferOp}
 O_{v \leftarrow u} =O_{u \leftarrow v}^{\sf T}:= O_v^e (O_u^e)^{\sf T}
\end{equation}
 satisfies $O_{v \leftarrow u} e_i^u = \pm e_i^v ,\;\; O_{u \leftarrow v} e_i^v = \pm e_i^u  \;\; \forall i.$ In other words:
\begin{equation}
O_{v \leftarrow u} b_u = b_v \text{ modulo signs,} \text{ and equivalently } O_{u \leftarrow v} b_v = b_u \text{ modulo signs}.
\end{equation}
\end{lemma}


\begin{proof}
First we prove the \textit{only if} condition. Given that $\beta_{uv}^{\sf T} e^u_i = \pm \sigma_i e^v_i$ and 
$ \beta_{uv} \beta_{uv}^{\sf T} e_i^u =  \sigma_i^2 e_i^u $ we have
\begin{equation*}
 O_v^e \Sigma_{uv}^{\rm SVD} (O_u^e)^{\sf T} e_i^u = \pm \sigma_i e_i^v\;\;,\;\; (\Sigma_{uv}^{\rm SVD})^2 (O_u^e)^{\sf T} e_i^u = \sigma_i^2 (O_u^e)^{\sf T} e_i^u.
\end{equation*}
If a real matrix $A$ is non-negative and diagonal, then any eigenvectors of $A^2$ with eigenvalues $\lambda$ are also the eigenvectors of $A$ with eigenvalues $\vert \sqrt{\lambda}\vert$. Since $\Sigma_{uv}^{\rm SVD}$ is non-negative and diagonal, we see that $(O_u^e)^{\sf T} e_i^u$ is an eigenvector of $\Sigma_{uv}^{\rm SVD}$ with eigenvalue $\vert \sigma_i \vert$. 
It follows that:
\begin{equation*}
O_v^e \Sigma_{uv}^{\rm SVD} (O_u^e)^{\sf T} e_i^u = \pm \sigma_i e_i^v  \Rightarrow
\vert \sigma_i \vert O_v^e(O_u^e)^{\sf T} e_i^u = \pm \sigma_i e_i^v
\end{equation*}
If $\sigma_i \neq 0$ then $O_v^e (O_u^e)^{\sf T} e_i^u =\pm   e_i^v$. Suppose $\sigma_i=0$, then since $\textrm{Rank}(\beta_{uv})>1$ there is a single $i$ for which this holds. Since for all $j\neq i$, $O_v^e (O_u^e)^{\sf T} e_j^u = \pm e_j^v$, it follows that $e_i^v$ must lie in the one-dimensional subspace spanned by $O_v^e (O_u^e)^{\sf T} e_i^u$, and so $O_v^e (O_u^e)^{\sf T} e_i^u = \pm e_i^v$.

Now we prove the \textit{if} condition. Given that $ \beta_{uv} \beta_{uv}^{\sf T} e_i^u =  \sigma_i^2 e_i^u $ and $O_v^e (O_u^e)^{\sf T} e_i^u  = \pm e_i^v$:
\begin{align*} 
 \beta_{uv} \beta_{uv}^{\sf T} e_i^u &=  \sigma_i^2 e_i^u \\
   O_u^e \Sigma^{\rm  SVD}_{uv} (O_v^e)^{\sf T}  O_v^e \Sigma^{\rm  SVD}_{uv} (O_u^e)^{\sf T} e_i^u &=  \sigma_i^2 e_i^u \\
(\Sigma_{uv}^{\rm SVD})^2 (O_u^e)^{\sf T} e_i^u &=  \sigma_i^2 (O_u^e)^{\sf T} e_i^u.
\end{align*}
Since $\Sigma_{uv}^{\rm SVD}$ is a positive diagonal matrix we have
\begin{align*} (\Sigma_{uv}^{\rm SVD}) (O_u^e)^{\sf T} e_i^u &=  \vert \sigma_i \vert (O_u^e)^{\sf T} e_i^u \\
 O_v^e(\Sigma_{uv}^{\rm SVD}) (O_u^e)^{\sf T} e_i^u &=  \vert \sigma_i \vert O_v^e (O_u^e)^{\sf T} e_i^u \\
 \beta_{uv}^{\sf T} e_i^u &= \pm \sigma_i  e_i^v.
\end{align*}
 By a symmetric argument  $ \beta_{uv} e_i^v = \pm \sigma_i e_i^u$.
\end{proof}
It is important to note that the construction of $O_{v \leftarrow u}$ in Eq.~(\ref{eq:transferOp}) is not unique. One could find a different singular value decomposition and construct a different operator $O_{v \leftarrow u}'$. However, as proven above, for any such operator its action on a singular vector $e_i^u$ of $\beta_{uv}$ is identical, up to a difference in the sign, which has no bearing on the problem. In light of this, for the remainder of the text we will treat the operator $O_{v \leftarrow u}$ as a well defined orthogonal operator, with the implicit assumption being that any such operator suffices. 

\iftrue

The above lemma has two important consequences.
\begin{corollary} \label{cor:path}\hfill
\begin{enumerate}
 \item Condition 2 in Definition \ref{def:ansatz} is equivalent to the condition that for every \highrank\ edge $(u,v)$ and any operator $O_{v \leftarrow u}$:
 $$ O_{v \leftarrow u}b_u = b_v \text{ modulo signs}.$$ 
\item If $B=\{ b_u\}$ is a candidate basis, then given a path $p= (u,x...,y,w,v)$ of \highrank\ edges going from vertex $u$ to $v$, for any orthogonal rotation defined by $O_{p} = O_{v \leftarrow w}O_{w \leftarrow y}... O_{x \leftarrow u}$, which we call a \textbf{\highrank\  path operator}, it must be the case that $ O_{p}  b_u= b_v$ modulo signs. 
\end{enumerate}
\end{corollary}

\begin{proof}
The first statement follows trivially by the definition of the candidate basis. The second statement follows by induction from the first. 
\end{proof}

\fi

We see that, if for some vertex $u$, we choose a basis $b_u$ which happens to belong to a yet unknown candidate basis $B$, then this fixes, via the \highrank\ path operators, all of the bases $b_v \in B$, modulo signs, for all vertices $v$ connected to $u$ by \highrank\ paths. This motivates the following definition.
\begin{definition}[\capHRCC\  (\abrvHRCC)] \label{def:HRCC}
Remove all \lowrank\ edges from the graph $G$. What remains is a family of distinct connected components which are composed entirely of \highrank\ edges. Define the \textbf{\HRCC}  as the subgraph $\Gamma$ associated with such a connected component. 
\end{definition}
Note that in the case where some vertex $v$ is connected to only \lowrank\ edges, $v$ on its own still constitutes a \HRCC . Therefore by construction every vertex is in exactly one \abrvHRCC. Note also that any two vertices connected by a path of \highrank\ edges belong to the same \abrvHRCC . 

\begin{definition}[Candidate Basis on a \abrvHRCC ]
Given a \abrvHRCC\ $\Gamma$, a \textbf{candidate basis of a \HRCC} $B_{\Gamma}$ on the vertices of $\Gamma$ is the assignment of a basis $b_u$ to each vertex $u \in \Gamma$ which satisfies all the conditions of a candidate basis for all of the \highrank\ edges in $\Gamma$, as well as for the \lowrank\ edges (which are not in $\Gamma$) which are adjacent to the vertices in $\Gamma$.
\end{definition}

Clearly, if we combine all the candidate bases for the RCCs $\Gamma$, we obtain a candidate basis for the vertices of the whole graph $G$. So the task of finding a candidate basis $B$ for the whole graph breaks up into finding a candidate basis $B_{\Gamma}$ for each \abrvHRCC\ $\Gamma$, so that $B = \bigcup_{\Gamma} B_{\Gamma}$. Furthermore, if we can correctly choose a basis $b_u$ at one vertex $u$ in $\Gamma$ then, due to Corollary~\ref{cor:path}, we will have successfully specified all of $B_{\Gamma}$. The primary challenge is making the right choice of $b_u$.

\begin{algorithm}
	\LinesNumbered
	\SetAlgoVlined
	
    \SetKwInOut{Input}{Input}
    \SetKwInOut{Output}{Output}

    \Input{Graph $G = (V, E)$, matrix edge weights $\{\beta_{uv}\}$, RCCs $\{\Gamma\}$ }
    \Output{A candidate basis $B= \{b_u\}$, if one exists. Otherwise False, indicating no candidate basis exists.}
	\SetKwBlock{Begin}{for $\Gamma \in \text{RCCs}$ do }{}
	\SetAlgoLined    
    \Begin{
    	\SetAlgoVlined
    	\tcc{Construct the $O_{v\leftarrow u}$ operators for all the edges in $\Gamma$}
    	\For{$e=(u,v)\in \Gamma$}{
    		$(O_u^e, \Sigma_{uv}^{\text{SVD}},  O_v^e) = \text{SVD}(\beta_{uv})$ s.t. 
    		$O_u^e \Sigma_{uv}^{\text{SVD}} (O_v^e)^T = \beta_{uv}$
    		
    		$O_{v\leftarrow u} = O_v^e (O_u^e)^{\sf T}$
    		
    		$O_{u\leftarrow v} =O_{v\leftarrow u}^{\sf T} $
    	}
    	
		\tcc{At each vertex take intersection of eigenspaces associated with immediately neighbouring edges, as in Algorithm \ref{alg:rank1ansatz}}    	
    	\For{$v \in \Gamma$}{

    		$\mathbb{S}_v[0] = \{ \mathbb{R}^3 \}$
    	
    	\For{$u \in V$ s.t. $e=(u,v) \in E$}{
    		$\mathbb{S}_v^e= $  the set of orthogonal maximal eigenspaces of $\beta_{uv}^{\sf T} \beta_{uv}$
    		
    		$\mathbb{S}_v[0] = \mathbb{S}_v[0] \cap \mathbb{S}_v^e$
    	}
    	}
    	
	\tcc{For each vertex, iteratively take the intersection of the subspaces at that vertex, with the appropriately rotated subspaces of the neighbouring vertices}    	
    	
    	i=0
    	
    	\While{True}{ \label{algLine:whileLoop}
    		\For{$v\in \Gamma$}{
    			$\mathbb{S}_{\text{neighbours}} = \{ \mathbb{R}^3\}$
    			
    			\For{$u \in \Gamma$ s.t. $(u,v) \in \Gamma$}{
    				$\mathbb{S}_{\text{neighbours}} =\mathbb{S}_{\text{neighbours}} \cap (O_{v \leftarrow u} \mathbb{S}_u[i]) $				
    			}
    			
    			 $\mathbb{S}_v[i+1] = \mathbb{S}_v[i] \cap\mathbb{S}_{\text{neighbours}}$\label{algLine:iterativeProcess}
    		}
    		\tcc{Conclude the iterative process when it reaches a fixed point. Note that since every subspace in $\mathbb{S}_v[i+1]$ is contained in a subspace of $\mathbb{S}_v[i]$, this process must reach a fixed point.}
			\If{$\mathbb{S}_v[i+1] = \mathbb{S}_v[i] \;\; \forall v\in \Gamma$}{
				$\mathbb{S}_v[f] := \mathbb{S}_v[i] \;\; \forall v \in \Gamma$
				
				break
			}    		
    		i++
    	}
    	
    	\tcc{Choose a spanning tree and construct the intersection of the eigenspaces of the \highrank\ path operators associated with the fundamental cycles. Take the intersection of this with the set of subspaces at the root vertex.}
    	
    	$T$ = spanning tree of $\Gamma$ with root vertex $r$.
    	
    	$\mathbb{S}_{\text{loops}} = \{\mathbb{R}^3\}$
    	
    	\For{edge $e \in \Gamma$ s.t. $e\not \in T$ \label{algLine:edgesNotInT}}{	
    	Get $C_e$, the fundamental cycle associated with $e$
   		 	
    	Let $p_e = (r,u,v,...,w,r)$ be the ordered vertex sequence of $C_e$
    	
		Construct $O_{p_e}= O_{r \leftarrow u} O_{u \leftarrow v}... O_{w \leftarrow r}$, the assocated \highrank\ path operator. \label{algLine:loopOp}
		
		Find $\{\lambda_i\}$ and $\mathbb{S}_{p_e} =\{S_i^{p_e}\}$, the eigenvalues and maximal eigenspaces of the orthogonal matrix $O_{p_e}$. (Note that every orthogonal matrix is diagonalizable) 	
    	
		\If{$\exists i \;\text{s.t.}\; \lambda_i\not \in \mathbb{R}$}{   
			return False  \label{algLine:imaginaryEigvals}	
    	}
    	$\mathbb{S}_{\text{loops}} = \mathbb{S}_{\text{loops}} \cap \mathbb{S}_{p_e}$

    }
	    
	$\mathbb{S}_r^* = \mathbb{S}_r[f] \cap\mathbb{S}_{\text{loops}} $
	
	\textbf{Continued below}
    }
    \caption{Algorithm for finding a candidate basis}\label{alg:genericAnsatz}
\end{algorithm}



\begin{algorithm}
  
      \SetKwInOut{Input}{}
      \SetKwProg{Def}{def}{:}{}
	\SetKwFunction{Propagate}{propagate}
    
\textbf{Algorithm \ref{alg:genericAnsatz}} continued
\SetKwBlock{Begin}{}{end}
\Begin{
	\Input{ }
	
    \If{
    	$\text{span}(\mathbb{S}_r^*) \neq \mathbb{R}^3$ 
    }{
    	return False \label{algLine:incompleteSpan}
    }
    \tcc{Note that the intersection of sets of orthogonal subspaces is also a set of orthogonal subspaces. Thus if $\text{span}(\mathbb{S}_r^*)=\mathbb{R}^3$ then one can always choose an orthogonal basis from it.}
    Choose orthonormal basis $b_r = (e_1^r, e_2^r, e_3^r) \subseteq \mathbb{S}_r^*$ \label{algLine:br}
 	
 	\tcc{propagate the choice of basis at the root vertex out to the rest of the vertices in the tree.}
    
	\Def{\Propagate{$u, b_u$}}{
	\For{vertex $v\in T$ that are children of $u$}{
		$b_v = O_{v\leftarrow u} b_u$ \label{algLine:propagate}
		
		\Propagate{$v$, $b_v$}
	}
	}    
    
    \Propagate($r$, $b_r$)
    
	$B_{\Gamma} = \{b_u\;:\; \forall u \in \Gamma \}$
  }
  return $B = \bigcup_{\Gamma} B_{\Gamma}$
\end{algorithm}

\begin{lemma}\label{lem:genericAnsatz}
Given a matrix weighted graph, Algorithm~\ref{alg:genericAnsatz} efficiently finds a candidate basis $B$ or otherwise shows that no such candidate basis exists. The algorithm takes $O(N^3)$ steps, where $N$ is the number of vertices.
\end{lemma}

\begin{proof}

Let us first prove that if the algorithm returns $B$, then $B$ is a candidate basis. To show that $B$ is a candidate basis, we need only show that for all $\Gamma$ $B_{\Gamma}$ is a candidate basis.

The first fact to note is that for every vertex $v \in \Gamma$, and for every edge $(v,w)$ adjacent to $v$, the basis vectors $b_v$ are eigenvectors of $\beta_{vw} \beta_{vw}^{\sf T}$ and so the first condition necessary for $B_{\Gamma}$ to be a candidate basis is satisfied. To see that this is true, it suffices to show that $b_v \in \mathbb{S}_v[f]$, since $\mathbb{S}_v[f] \subseteq \mathbb{S}_v[0]$, and $\mathbb{S}_v[0]$ by construction only contains eigenvectors of all neighbouring edges, including \lowrank\ edges. For all $w \in \Gamma$, since $\mathbb{S}_w[f]$ is a fixed point of 
 the equation on Line (\ref{algLine:iterativeProcess}) of Algorithm \ref{alg:genericAnsatz},
  we have
\begin{equation}
\mathbb{S}_w[f] = \mathbb{S}_w[f] \cap \left( \bigcap_x O_{w \leftarrow x} \mathbb{S}_x[f] \right),
\end{equation}
 where $x$ runs over \highrank\ edges adjacent to $w$. Thus for all \highrank\ edges $(w,x)$ in $\Gamma$,  $\mathbb{S}_w[f] \subseteq O_{w \leftarrow x} \mathbb{S}_x[f]$ . Consequently, given a vertex $w$ in the spanning tree $T$, and a child vertex $x$, if $b_w \subseteq \mathbb{S}_w[f]$, then since $b_x = O_{x \leftarrow w} b_w \Rightarrow O_{w \leftarrow x} b_x = b_w$ it follows that  $b_x \subseteq \mathbb{S}_x[f]$. Since $b_r \subseteq \mathbb{S}_r[f]$, and $r$ is the root node of $T$, it follows by induction that for all $v \in \Gamma$:  $b_v \subseteq \mathbb{S}_v[f]$. 
 

The second fact to note is that for every \highrank\ edge $(w,v)$ in $\Gamma$, $b_w = O_{w \leftarrow v} b_v$, modulo signs. This is clearly true for every \highrank\ edge in $T$ by construction, as specified in 
Line (\ref{algLine:propagate}) of Algorithm \ref{alg:genericAnsatz}.
All that remains are those \highrank\ edges not in $T$.  Consider an edge $e = (v,w)$ not in $T$.  There is a fundamental cycle $C_e$, with a path $p_e$ which goes from the root vertex $r$, up to $v$, entirely along paths in the spanning tree, then from $v$ to $w$, and then from $w$ back to $r$. Thus the associated \highrank\ path operator is
\begin{equation*}
 O_{p_e} = O_{r \leftarrow x }... O_{y \leftarrow w} O_{w \leftarrow v} O_{v \leftarrow z} ... O_{q \leftarrow r}.
\end{equation*}
Furthermore, the bases $b_v$ and $b_w$ are, by construction:
\begin{equation*}
 b_v = O_{v \leftarrow z} ... O_{q \leftarrow r} b_r
 \end{equation*}
 \begin{equation*}
  b_w =  O_{w \leftarrow y}...O_{x \leftarrow r} b_r  \rightarrow b_w =  O_{y \leftarrow w}^{\sf T}...O_{r \leftarrow x}^{\sf T} b_r 
 \end{equation*}
By construction, every element in $b_r$ must be an eigenvector of $O_{p_e}$ with real eigenvalues (equal to $+1$ or $-1$ since $O_{p_e}$ is an orthogonal matrix)
(see Line (\ref{algLine:br}) of Algorithm \ref{alg:genericAnsatz}).
  Thus $ b_r=O_{p_e} b_r  $  modulo signs. Therefore
\begin{align*}
b_r &= O_{r \leftarrow x }... O_{y \leftarrow w} O_{w \leftarrow v} O_{v \leftarrow z} ... O_{q \leftarrow r} b_r \text{ modulo signs}, \\
O_{y \leftarrow w}^{\sf T}...O_{r \leftarrow x}^{\sf T}  b_r &= O_{w \leftarrow v} O_{v \leftarrow z} ... O_{q \leftarrow r} b_r \text{ modulo signs}, \\
b_w &= O_{w \leftarrow v} b_{v}  \;\; \textrm{ modulo signs}.
\end{align*} 
Thus for every \highrank\ edge $(w,v)$ in $\Gamma$, $b_w = O_{w \leftarrow v} b_v$, modulo signs. Combining this fact with  
claim 1 of Corollary~\ref{cor:path}, it is clear that the second condition  necessary for $B_{\Gamma}$ to be a candidate basis is satisfied. Therefore $B_{\Gamma}$ is a candidate basis.

\bigbreak
Now we will prove that if Algorithm \ref{alg:genericAnsatz} returns False, then no candidate basis exists. First we note that obviously if for any $\Gamma$ there does not exist a candidate basis $B_{\Gamma}$, then no candidate basis exists for the whole graph. 

There are two places where the algorithm returns False. Once at Line (\ref{algLine:imaginaryEigvals}), and once at Line (\ref{algLine:incompleteSpan}). This happens at Line (\ref{algLine:imaginaryEigvals}) if some $O_{p_e}$ has any non-real eigenvalues. Note that by claim 2 in Corollary~\ref{cor:path}, if there existed a candidate basis $B=\{ b_u\}$, then $O_{p_e} b_r = b_r$, modulo signs, since $O_{p_e}$ is a \highrank\ path operator. In other words, the eigenvalues of $O_{p_e}$ should be either $+1$ or $-1$ \footnote{For an example of where such a loop operation becomes important, see Appendix~\ref{app:cliffordsNotSuffice}}. 

The algorithm indicates at Line (\ref{algLine:incompleteSpan}) that no candidate basis exists if, for a given \abrvHRCC\ $\Gamma$ with root vertex $r$, $\textrm{span}\left( \mathbb{S}_r^* \right) \neq \mathbb{R}^3$. This happens if and only if there does not exist a set of three orthogonal vectors such that each of them belongs to a subspace in $\mathbb{S}_r[f]$ as well as a subspace in every $\mathbb{S}_{p_e}$. We prove by contradiction that in this case $B_{\Gamma}$ must not exist.

Suppose there does not exist a set of three orthogonal vectors such that each of them belongs to a subspace in $\mathbb{S}_r[f]$ as well as a subspace in every $\mathbb{S}_{p_e}$. Suppose a candidate basis $B_{\Gamma}$ does exist, then by the argument made for the case of Line (\ref{algLine:imaginaryEigvals}) in the preceding paragraph, the basis vectors $b_r$ must be eigenvectors of every $O_{p_e}$ and thus each vector in $b_r$ belongs to a subspace in $\mathbb{S}_{p_e}$ for every $p_e$. Therefore it must be that $b_r \not \subseteq \mathbb{S}_r[f]$. However this is contradicted by the following argument.

First note that for all $b_v \in B_{\Gamma}$ and for every adjacent edge $(v,x)$, the basis vectors $b_v$ must be eigenvectors of $\beta_{vx} \beta_{vx}^{\sf T}$, and thus $b_v \subseteq \mathbb{S}_v[0]$. Second note that if for all $ b_v \in B_{\Gamma}: \; b_v \subseteq \mathbb{S}_v[i]$, then for all $b_v \in B_{\Gamma}: \; b_v \subseteq \mathbb{S}_v[i+1]$.  This follows from the fact that for all vertices $v \in \Gamma$, and for all \highrank\ edges $(v,x)$ adjacent to $v$, $b_v  = O_{v \leftarrow x} b_x$, modulo signs, by Corollary~\ref{cor:path}, and thus $b_v \subseteq O_{v \leftarrow x}\mathbb{S}_x[i]$. Since $\mathbb{S}_v[i+1] = \mathbb{S}_v[i] \cap \left(\bigcap_{x} O_{v \leftarrow x}\mathbb{S}_x[i] \right) $ it follows that $b_v \subseteq \mathbb{S}_v[i+1]$. Thus, by induction,  $\forall b_u \in B_{\Gamma}: \; b_u \subseteq \mathbb{S}_u[f]$, which is a contradiction. 

\bigbreak


Finally we prove that the algorithm runs in $O(N^3)$ steps, where $N$ is the number of vertices in the graph. The most costly part of the algorithm is the while loop at Line (\ref{algLine:whileLoop}). Let $n_{\Gamma}$ be the number of vertices in the \abrvHRCC\ $\Gamma$. Constructing all $\mathbb{S}_u[0]$ runs in worst case $O(n_{\Gamma}N)$. Each iterative step runs in worst case $O(n_{\Gamma}N)$. At each iterative step the subspaces of $\mathbb{S}_u[i+1]$ must be contained in the subspaces in $S_u[i]$. Therefore if we have not reached a fixed point, then at each iterative step there is at least one $\mathbb{S}_u[i+1]$ for which the dimensions of the subspaces have decreased when compared to $\mathbb{S}_u[i]$. If $\mathbb{S}_u[i]$ spans $\mathbb{R}^3$, then the dimensions of its mutually orthogonal subspaces must be either $(3)$, $(2,1)$ or $(1,1,1)$. Thus for every vertex $u$, the iterative process can only decrease the dimensionality of the subspaces in $\mathbb{S}_u[i]$ at most three times before $\mathbb{S}_u[i]$ no longer spans $\mathbb{R}^3$. So the maximum number of iterations is $3n_{\Gamma}$. 

Therefore the naive worst case runtime of this step is $O(n_{\Gamma}^2N)$. However we expect that a more careful analysis would find the runtime to be closer to $O(n_{\Gamma}N)$, since the runtime of each iterative step is proportional to the connectivity of the graph, while the number of iterative steps required should be inversely proportional to the connectivity.

At Line (\ref{algLine:edgesNotInT}) the algorithm iterates over edges in $\Gamma$ not in $T$, the number of which is upper bounded by $O(n_{\Gamma}^2)$. All other steps in the algorithm iterate over vertices in $\Gamma$, or edges adjacent to those vertices, and so have runtime at most $O(n_{\Gamma}N)$.

Since the whole algorithm iterates over all \abrvHRCC s, it follows that the  runtime is $O \left( \sum_{\Gamma} n_{\Gamma}^2 N \right)$ which, by the triangle inequality, is upper bounded by $O(N^3)$.
\end{proof}




\begin{algorithm}
    \SetKwInOut{Input}{Input}
    \SetKwInOut{Output}{Output}

    \Input{Graph $G = (V, E)$, \highrank\ and \lowrank\ matrix edge weights $\{\beta_{uv}\}$, RCCs $\{\Gamma\}$, and candidate basis $B=\{b_u\}$ }
    
        \Output{A set of permutations $\Pi = \{\pi_u\}$ such that $B^{\Pi}$ is an NLY basis, if one exists. Otherwise False, indicating none exists.}
	\For{$\Gamma \in $ RCCs}{
		\tcc{Label all \lowrank\ edges incident to vertices $v$ in $\Gamma$ according to which basis vector $b_v$ is not in the null space of $\beta_{uv}$. These always exist by Proposition \ref{prop:lowrank}}
		$L(\Gamma) = \{\}$
		
		\For{$e =(u,v) \in E$ s.t. $v \in \Gamma$, and $\text{rank}(\beta_{uv})=1$ }{		
		\For{$i\in\{1,2,3\}$}{
			$e_i^u = b_u[i]$
			
			\If{$\beta_{uv}e_i^v \neq 0$}{
			
				$L(v, e)=i$
				
				$L(\Gamma) = L(\Gamma) \cap \{i\}$
			}			
		}
		
		}
		\tcc{If $\Gamma$ is incident on more than two different labels, then return false.}
		\If{$\vert L(\Gamma)\vert=3$}{
			return False
		}
		\tcc{Define permutations so that all incident edge labels are mapped to either $1$ or $3$}
		
			Choose perm. $\pi_{\Gamma}^0$ s.t. $\pi_{\Gamma}^0(L(\Gamma)[1])=1$, and if $\vert L(\Gamma)\vert>1$ then $\pi_{\Gamma}^0(L(\Gamma)[2])=3$
			
			Choose perm. $\pi_{\Gamma}^1$ s.t. $\pi_{\Gamma}^1(L(\Gamma)[1])=3$, and if $\vert L(\Gamma)\vert>1$ then $\pi_{\Gamma}^1(L(\Gamma)[2])=1$ 
 		}
 		
 		\tcc{The task now becomes assigning a binary value $x_{\Gamma}$ to each $\Gamma$ so that for every \lowrank\ edge $e=(u,v)$, with labels $L(u,e)=i$ and $L(v,e)=j$, which has vertices in $\Gamma_u$ and $\Gamma_v$, the binary assignment $x_{\Gamma_u}$ to $\Gamma_u$ and $x_{\Gamma_v}$ to $\Gamma_v$ satisfy: $\pi^{x_{\Gamma_u}}_{\Gamma_u}(i) = \pi^{x_{\Gamma_v}}_{\Gamma_v}(j).$}
 		
		\tcc{For each \lowrank\ edge define a 2-XOR-SAT clause on boolean variables associated with the RCCs on which the edge is incident.} 		
 		
 		\For{$e=(u,v)\in E$ s.t. $\text{rank}(\beta_{uv})=1$}{
 			Find $\Gamma_u\in $ RCCs s.t. $u\in \Gamma_u$. 
 			
 			Find $\Gamma_v \in $ RCCs s.t. $v \in \Gamma_v$.

 			\eIf{$L(v,e) = L(u,e)$}{
 				$C_e(x_{\Gamma_u},x_{\Gamma_v}) = x_{\Gamma_u} \oplus x_{\Gamma_v}$
 			}{
				 $C_e(x_{\Gamma_u},x_{\Gamma_v}) = \bar{ x}_{\Gamma_u} \oplus x_{\Gamma_v}$			
 			}
 		}
 		\tcc{The solution to the associated 2-XOR-SAT problem specifies, for each RCC $\Gamma$, which permutations to apply in a uniform fashion to all vertices in $\Gamma$.} 
 		
 		variables = $\{ x_{\Gamma} \;\vert\; \forall \Gamma \in \text{RCCs}\}$
 		
 		clauses =  $\{C_e \;\vert\; \forall e=(u,v) \in E \text{ s.t. } \text{rank}(\beta_{uv})=1 \}$
 		
 		success= 2-XOR-SAT(ref variables, clauses)		
 		
		\If{$\lnot$success}{
			return False
		}

 		\For{$\Gamma \in $ RCCs}{
 			\For{vertex $v \in \Gamma$}{
				 $\pi_v = \pi_{\Gamma}^{x_{\Gamma}}$
 			}
 		}
 		
 		return $\Pi = \{ \pi_v \;\vert\; \forall v \in V\}$

    \caption{Algorithm for finding permutations $\Pi$ such that $B^{\Pi}$ is a NLY basis}\label{alg:genericPerms}
\end{algorithm}

\begin{theorem}\label{thrm:genericAlgo}
Given a matrix weighted graph, one can efficiently find an NLY basis, or else show that no such basis exits.
\end{theorem}

\begin{proof}
The procedure for finding an NLY basis is to first find a candidate basis $B$ using Algorithm~\ref{alg:genericAnsatz}, and then find a set of permutations $\Pi$ such that $B^{\Pi}$ is an NLY basis using Algorithm~\ref{alg:genericPerms}. 

If Algorithm~\ref{alg:genericPerms} is succesful, then $B^{\Pi}$ is an NLY basis by the following reasoning. For every \highrank\ edge an identical permutation is applied to its adjacent vertices, and so the diagonality of the matrix weights is preserved. While for the \lowrank\ edges, the matrix weights are diagonal, and by construction every matrix weight is zero in its second diagonal entry, as per arguments made in subsection \ref{subsec:rank1graph}. 

Algorithm~\ref{alg:genericPerms} is efficient, since the number of variables  in the 2-XOR-SAT problem is the number of \abrvHRCC s, and in the worst case this is the number of vertices $N$, so it runs in time $O(N^2)$. Combining this with Lemma~\ref{lem:genericAnsatz} the worst case runtime of the whole algorithm is $O(N^3)$.

Finally, if either of these algorithms fail, then we claim that no NLY basis exists, by the following two arguments. Firstly, if Algorithm~\ref{alg:genericAnsatz} fails, then by Lemma~\ref{lem:genericAnsatz} no candidate basis exists, and since an NLY basis must satisfy the conditions for being a candidate basis, no NLY basis exists. Secondly, we must establish the non-trivial fact that if Algorithm~\ref{alg:genericPerms} fails, then no NLY basis exists. In other words, we need to rule out the possibility that Algorithm~\ref{alg:genericPerms} might have succeeded had we supplied it with an alternative candidate basis. The rest of our exposition is devoted to proving this fact.

First note that, when given a candidate basis $B$, if Algorithm~\ref{alg:genericPerms} fails then no set of permutations exists such that $B^{\Pi}$ is an NLY basis. This follows from the fact that, if a permutation were to exist, it must be uniform on every \abrvHRCC , in order to preserve the diagonality of the \highrank\ edges. Given this, the argument reduces to the same one made in Theorem~\ref{thrm:rank1algo}, where we treat \abrvHRCC s as sites, since every \lowrank\ edge is adjacent to \abrvHRCC s. 

Given that if the procedure in Algorithm~\ref{alg:genericPerms} fails, then there does not exist a permutation $\Pi$ such that the basis $B^{\Pi}$ is an NLY basis, we use proof by contradiction to show that in this case no NLY basis exists.

 Suppose there exists an NLY basis $\bar{B} = \{ \bar{b}_u\}$. We now argue that for a fixed RCC $\Gamma$, for every index $i \in \{1,2,3\}$ there exists an index $j \in \{1,2,3\}$ such that for every vertex $u \in \Gamma$, if $e_i^u \in b_u$ corresponds to a left singular vector, with non-zero sigular value, of a \lowrank\ edge adjacent to $u$, then $e_i^u = \pm \bar{e}_j^u \in \bar{b}_u$. 

Given an index $i$, consider any two vertices $u,v \in \Gamma$ for which $e_i^u, e_i^v$ correspond to singular vectors, with non-zero singular values, of some \lowrank\ edges adjacent to $u$ and $v$. Consider that $\bar{b}_u$ must also contain a vector $\bar{e}_{j_u}^u$ at some particular index $j_u$, which is also a singular vector with non-zero singular value of the same \lowrank\ edge adjacent to $u$. Since that edge is \lowrank , it follows that $e_i^u = \pm \bar{e}_{j_u}^u$. A similar argument can be made for $v$ so that $e_i^v = \pm \bar{e}_{j_v}^v$, for some index $j_v$.  There must exist a \highrank\ path $p$ connecting $u$ to $v$,  and by Corollary~\ref{cor:path}, $O_p e_i^u = \pm e_i^v$. Similarly, $O_p \bar{e}_{j_u}^u = \pm \bar{e}_{j_u}^v$.  Therefore $\pm  \bar{e}_{j_u}^v =\bar{e}_{j_v}^v$ and since each vector in $\bar{b}_v$ is orthogonal we have $j_u =j_v =j$. Since this is true for any such pair of vertices $u,v$, there must exist an index $j$ such that for every vertex $u \in \Gamma$, if $e_i^u$ corresponds to a left singular vector, with non-zero sigular value, of a \lowrank\ edge adjacent to $u$, then $e_i^u = \pm \bar{e}_j^u \in \bar{b}_u$.

We can now proceed by the same reasoning used in the proof of Theorem~\ref{thrm:rank1algo}. Let $\pi_{\Gamma}$ be the permutation with the mapping: $\pi_{\Gamma}(i)=j$, and let $\pi_u = \pi_{\Gamma}$ for all $u\in \Gamma$, and define $\Pi =\{\pi_u\}$. Then the bi-labelled graph associated with $B^{\Pi}$ must be identical to the bi-labelled graph associated with $\bar{B}$ (i.e. the action on the \lowrank\ edges is the same), and therefore $B^{\Pi}$ must be an NLY basis, which is a contradiction.
\end{proof}

\section{Discussion}
\label{sec:dis}


It is clear from the work presented here that in the case of two-local qubit Hamiltonians, the hardness of curing the sign problem by local basis transformations is determined by the presence or absence of one-local terms in the Hamiltonian. 


The question of whether the general {\sf LocalSignCure} is a problem in NP for general two-local $n$-qubit Hamiltonians is not clear, as the set of local unitary transformations is a continuous parameter space, and a prover would need to specify a sign-curing solution with a polynomial number of bits and such exact sign-curing transformations may not exist. A natural relaxation would be to demand that the transformation be approximately sign-curing, a direction of research that is explored in Ref.~\cite{HRNE:easing}, so that the problem would be contained in MA. 

We do not know the complexity of determining the ground state energy for the family of Hamiltonians presented in section 5. It may be that finding the ground state energy is easy, thus obviating any interest in a curing transformation. It would be interesting to show that a family of Hamiltonians exists for which deciding sign-curing and determining ground state energy are both hard. It would be very surprising indeed if the hardness of deciding a curing transformation only appears when the ground state energy is efficiently computable.

A natural extension of sign-curing transformations beyond single-qubit unitary transformations are transformations which first embed each qubit into a $d$-dimensional system, and then allow for local basis changes in this $d$-dimensional system. The power of such ``lifting" basis changes is completely unexplored, even in the two-qubit case. Another class of sign-curing transformations are Clifford circuits which map a Hamiltonian composed of ${\rm poly}(n)$ $k$-local Pauli's onto a sum of ${\rm poly}(n)$ non-local Pauli's. The power of these transformations is also largely unexplored, but some first results are reported in \cite{thesis:ioannou}. 
Recently, it was demonstrated~\cite{hen2019}  that even when there is an essential sign problem in the Hamiltonian, there are ways to group terms in the expansion of the Gibbs state to avoid the sign problem. It would be interesting to understand better how these techniques relate to stoquastic Hamiltonians. 

Another strand of interesting future research concerns the distinction between termwise and globally stoquastic Hamiltonians. Examples can be constructed of 3-local globally-stoquastic but not termwise-stoquastic Hamiltonians and the complexity of deciding global stoquasticity can be analyzed.

\appendix

\section{A simple example of non-stoquastic two-local Hamiltonian}\label{app:simpleExample}
Here we present a two-qubit Hamiltonian that can not be transformed into a symmetric $Z$-matrix by any single-qubit unitary transformations. Consider the Hamiltonian
\begin{equation*}
 H = -ZZ - 2XX + 3YY + IX + IZ + ZI + XI .
 \end{equation*}
The $\beta$-matrix of this Hamiltonian is of the form:
\begin{equation*}
 \beta = \left( \begin{matrix}
-2 &0&0\\
0& 3 & 0 \\
0 & 0& -1 
\end{matrix} \right).
\end{equation*}
First note that this Hamiltonian is not stoquastic in this basis because $ a_{XX} >  - \vert a_{YY} \vert$.  The orthogonal rotations on the $\beta$-matrix must be confined to the $XZ$ plane, in order to avoid complex terms like $IY$ or $XY$. Any pair of such orthogonal transformations (given by angles $\theta_1$ and $\theta_2$) will keep the $\beta$-matrix into a block-diagonal form, and the new $a_{XX}'$ entry will be
\begin{equation*}
a_{XX}' = -2 \cos(\theta_1) \cos(\theta_2) - \sin(\theta_1) \sin(\theta_2) > -3.
\end{equation*} 
Therefore $ a_{XX}' >  - \vert a_{YY} \vert$ for all values of $\theta_1$ and $\theta_2$, and so $H$ can not be transformed into a symmetric $Z$-matrix by single-qubit unitary transformations.

\section{Curing the sign problem for strictly two-local Hamiltonians by single-qubit Clifford transformations is easy}\label{app:cliffordsEasy}

Suppose that instead of single-qubit unitaries, one is interested in curing a strictly two-local Hamiltonian by single-qubit Clifford transformations. It is straightforward to show, by similar arguments to those outlined in Section~\ref{sec:curingResult} that such a problem is easy as we will do below. 

First, we know that the transformations that are employed in the XYZ-algorithm are single-qubit Clifford transformations. Furthermore, single-qubit Clifford transformations correspond to signed permutations on the matrix-weighted graph. Therefore, by the same logic as for the one-local unitary case, it suffices to find an algorithm which answers Problem ~\ref{problemStatement}, where instead of searching for a set of orthogonal rotations $\{O_u\}$, one instead seeks a set of signed permutations $\{\Pi_u\}$.

Now, instead of being able to consider any basis $B = \{b_u\}$, we only consider bases which are related to the standard basis by signed permutations:
\begin{equation*}
b_u =\{\Pi_u e_1, \Pi_u e_2, \Pi_u e_3\}.
\end{equation*}
Again, the signs are irrelevant, and therefore we are looking for bases which are related to the standard basis by a permutation.

We say a matrix is quasi-monomial if for each row and column of that matrix there is at most one non-zero entry \footnote{In a monomial matrix each row and column have exactly one non-zero entry, hence the addition `quasi'.}. A matrix $\beta_{uv}$ for an edge $e=(u,v)$ which is quasi-monomial admits a singular value decomposition of the form:
\begin{equation*}
\beta_{uv} = \Pi_u \Sigma_{uv}^{\rm SVD} (\Pi_v)^{\sf T},
\end{equation*}
where $\Pi_u$ and $\Pi_v$ are signed permutations. This can be seen by noting that by an appropriate permutation of columns and rows, the non-zero entries of $\beta_{uv}$ can be made positive and put on the diagonal in descending order, which corresponds to the singular value decomposition of $\beta$. It follows by definition that any $O_{v\leftarrow u}$, as defined in Eq.~(\ref{eq:transferOp}), is also a signed permutation. It is not difficult to see that if any edge in the matrix-weighted graph of our Hamiltonian has a weight $\beta_{uv}$ which is not quasi-monomial, then there can be no set of signed permutations which simultaneously diagonalize the weights of the graph. This is because a diagonalized matrix is quasi-monomial, and it is impossible to transform a non-quasi-monomial matrix into a quasi-monomial matrix by permuting the rows and columns.

We now describe the algorithm for answering problem statement~\ref{problemStatement} in the case where we are interested in signed permutations instead of orthogonal rotations. First check that every matrix weight $\beta_{uv}$ is a quasi-monomial matrix as this is a necessary condition by the arguments above. If any of these matrices are not, then we return false. 

We then identify a candidate basis $B_{\Gamma}$ for each \HRCC\ such that that the bases $b_u \in B_{\Gamma}$ are permutations of the standard basis. For each \HRCC\ construct the set of subspaces $\mathbb{S}_u^*$ as described in algorithm~\ref{alg:genericAnsatz}. Then check if the standard basis belongs to $\mathbb{S}_u^*$. By the same arguments made in~\ref{lem:genericAnsatz} it is clear that if the standard basis is not in $\mathbb{S}_u^*$ then no permutations of the standard basis are in $\mathbb{S}_u^*$, and so no candidate basis exists for $\Gamma$ which is a permutation of the standard basis, and so we must return false. If the standard basis is in $\mathbb{S}_u^*$, then we choose the standard basis for $b_u$, and construct a candidate basis $B_{\Gamma}$ for $\Gamma$ as per step 8 of Algorithm~\ref{alg:genericAnsatz}. Since every operator $O_{v \leftarrow u}$ is a signed permutation, it follows that all other bases $b_v \in B_{\Gamma}$ are permutations of the standard basis. Let $B = \bigcup_{\Gamma} B_{\Gamma}$.

Equipped now with a candidate basis for the graph, we can proceed with Algorithm~\ref{alg:genericPerms}. Noting that the only transformations being performed in this section are permutations on the candidate basis, we know that any NLY basis that is found will be a permutation of the standard basis. As such, whichever answer it gives will be the answer to our problem. $\Box$

We remark that even when the graph is weighted by only quasi-monomial matrices, it is generally not sufficient to consider only single-qubit Clifford transformations as curing transformations. This is proved in Appendix~\ref{app:cliffordsNotSuffice}.

\section{Single-qubit Clifford transformations do not suffice to cure the sign problem for a quasi-monomial matrix weighted graph}\label{app:cliffordsNotSuffice}
In this Appendix we show that when the $\beta$-matrices associated with a graph are quasi-monomial, as introduced in Appendix~\ref{app:cliffordsEasy}, then, even if there does not exist a set of signed permutations $\{ \Pi_u\}$ such that $\Pi_u^{\sf T} \beta_{uv} \Pi_v $ is diagonal, there may still exist a set of orthogonal transformations $\{O_u \}$ such that $O_u^{\sf T} \beta_{uv} O_v $ is diagonal. This is in contrast to the XYZ-algorithm. In the XYZ-algorithm all $\beta$-matrices are diagonal, a subclass of quasi-monomial matrices. In that case it was shown in Ref.~\cite{klassen2018two} that if there does not exist a set of signed permutations $\{ \Pi_u\}$ such that $\Pi_u^{\sf T} \beta_{uv} \Pi_v $ is diagonal, then there also does not exist a set of orthogonal transformations $\{O_u \}$ such that $O_u^{\sf T} \beta_{uv} O_v $ is diagonal, and so it is sufficient to consider signed permutations.

This insight is somewhat surprising for the following reason. If one considers a single quasi-monomial matrix $\beta$, it holds that all other quasi-monomial matrices $\beta'$ which can be obtained by orthogonal  transformations $O_1^{\sf T} \beta O_2 = \beta'$ can also be obtained by signed permutations $\Pi_1^{\sf T} \beta \Pi_2 = \beta'$ . One can see this by noting that the absolute values of the non-zero entries of $\beta$ are its singular values, and so the singular value decomposition of $\beta$ is related to all of the quasi-monomial matrices by the shuffling and flipping the signs of the rows and columns. 



Consider a matrix-weighted graph, whose matrix weights are monomial matrices, in particular consider a triangle with three qubits and Hamiltonian of the form:
\begin{equation*}
 H = H_{12}+ H_{23} + H_{31} \; \; ,\;  H_{uv} = X_u Y_v + Y_u X_v. 
\end{equation*}
The corresponding matrix weights of our graph are thus of the form:
\begin{equation*}
 \beta_{uv} = \left(\begin{array}{ccc}
0 & 1 & 0 \\
1 & 0 & 0 \\
0 & 0 & 0 
\end{array}\right).
\end{equation*}
It is not hard to see that no permutations exist which simultaneously diagonalize all three matrices. 
However, if we apply the rotation \begin{equation}\label{eq:topoSol}
O = \frac{1}{\sqrt{2}} \left(\begin{array}{ccc} 
1 & -1 & 0 \\ 
1 & 1 & 0 \\
0 & 0 & 1 
\end{array}\right),\end{equation} at every vertex, then every matrix is diagonalized. Translating back to the application of basis changes, this transformation corresponds to applying T-gates, of the form $\left(\begin{array}{cc}1 & 0 \\ 0 & e^{i \pi/4}  \end{array}\right)$, on all the qubits.

We see that in the case of quasi-monomial matrix-weighted graphs, the family of graphs which are equivalent under orthogonal transformations are not equivalent under signed permutations, and instead form sectors which depend on the graph topology.  It is precisely the \highrank\ loops considered in our algorithm in section~\ref{sec:curingResult} which captures this non-trivial topological structure. In Algorithm~\ref{alg:genericAnsatz}, one needs to check the loop operators $O_{P_e}$ (Algorithm~\ref{alg:genericAnsatz} Line~\ref{algLine:loopOp})) in order to identify the rotation~\ref{eq:topoSol}.


\section{Heisenberg Form Under Local Unitaries is NP-hard}\label{app:HeisForm}
In this appendix we argue that deciding if a 2-local Hamiltonian can be put into \textit{Heisenberg form} by single qubit unitary rotations is NP-hard. Equivalently, for the case of a matrix weighted graph of the kind considered in Section~\ref{sec:curingResult}, finding a set of orthogonal rotations $\{O_u\}$ that diagonalize every matrix weight is NP-hard.

\begin{definition}[Heisenberg Form]
A Hamiltonian of the form:
$$ H = \sum_{\langle u v \rangle} J_x^{uv} X_uX_v +J_y^{uv} Y_u Y_v +J_z^{uv} Z_u Z_v $$
is said to be in \textit{Heisenberg form}.
\end{definition}

\begin{definition}[Cubic Graph]
A cubic graph is a graph in which all vertices have degree three.
\end{definition}

\begin{definition}[Chromatic Index]
The chromatic index of a graph is the minimum number of colours required to color the edges of the graph in such a way that no two adjacent edges have the same colour.
\end{definition}

\begin{theorem}[\cite{Holyer81}]
It is NP-complete to determine whether the chromatic index of a cubic graph is 3 or 4.
\end{theorem}

\begin{corollary}[Heisenberg Form is NP-hard]
It is NP-hard to determine if a 2-local Hamiltonian can be put into Heisenberg form by applying single qubit unitary rotations.
\end{corollary}

\begin{proof}
The proof proceeds by reduction.

Consider a cubic graph $G=(E,V)$. Consider a two-local Hamiltonian $H_G$, acting on qubits associated with the vertices $V$. For every edge $e=(u,v) \in E$ $H_G$ consists of a two local term $P_uQ_v\;,\; P,Q \in \{X,Y,Z\}$, with the restriction that $\forall u \in V$: $X_u$, $Y_u$ and $Z_u$ appear in exactly one such term. This last restriction is always possible because every vertex has degree 3.

If $G$ has chromatic index $3$, then $H_G$ can be put into Heisenberg form by single qubit unitaries. Let the colouring be given by the Pauli operators $X$, $Y$, and $Z$. For every edge $e=(u,v)$ coloured by the Pauli operator $R$, map the Hamiltonian term $P_uQ_v$ to the term $R_uR_v$. Since every term in $H_G$ acting on a given qubit acts with a different Pauli operator, and every edge incident on a vertex is given a different colour, each Pauli operator acting on a given qubit is mapped to a unique Pauli operator, and so such a mapping can be given in terms of single qubit Clifford operators.

For the reverse direction, we prove that if $H_G$ can be put into Heisenberg form $H'_G$ by single qubit unitary rotations, then $G$ has chromatic index $3$.
Putting $H_G$ into Heisenberg form by single qubit unitary rotations is equivalent to diagonalizing all $\beta_{uv}$ matrices by orthogonal transformations, as defined in \ref{prop:ortho}, for all edges $(u,v)$. Consider that for a given edge $(u,v)$, the matrix $\beta_{uv}$ is rank-1. Thus for every edge $(u,v)$, $H'_G$ contains a single two qubit term of the form $X_uX_v$, $Y_uY_v$ or $Z_uZ_v$. Furthermore, all edges incident on a vertex must be associated with a different Pauli operator, since the transformation is unitary. Thus $H'_G$ prescribes a $3$ colouring, and $G$ has chromatic index $3$.
\end{proof}

Note that the above proof applies regardless of if the Hamiltonians are restricted to being exactly 2-local or not. Note also that problems of this type are ruled out by the No-Lone-YY condition introduced in Section~\ref{sec:curingResult}.

\section*{Acknowledgments}
The authors would like to thank Sergey Bravyi and Daniel Lidar for their thoughtful comments on this work.

\bibliographystyle{IEEEtran}
\bibliography{refs_stoq}
\end{document}